%% file: LQP_LMI_FH_arxiv.tex
\documentclass[10pt,fleqn]{article}

\input{common_commands}

\theoremstyle{definition}
\setlist[itemize]{leftmargin=*}

\renewcommand{\calQ}{{\bm{Q}}}

\title{Linear-Quadratic Problems in Systems and Controls \\ 
 via Covariance Representations and  Linear-Conic Duality: Finite-Horizon Case }           

\author{Bassam Bamieh\thanks{Department of Mechanical Engineering, UCSB, {\em bamieh@ucsb.edu}. 
	}} 

\date{}

\begin{document} 

\maketitle


\begin{abstract} 	
	Linear-Quadratic (LQ) problems that arise in systems and controls include the classical optimal control problems of the Linear Quadratic Regulator (LQR) in both its deterministic and stochastic forms, as well as $\sfH^\infty$-analysis (the Bounded Real Lemma), the Positive Real Lemma, and general Integral Quadratic Constraints (IQCs) tests. We present a unified   treatment of all of these problems using an approach which converts  linear-quadratic problems to matrix-valued linear-linear problems with a positivity constraint. This is done through a system representation where the joint state/input covariance (the outer product in the deterministic case) matrix is the fundamental object. LQ problems then become infinite-dimensional semidefinite programs, and the key tool used is that of linear-conic duality. Linear Matrix Inequalities (LMIs) emerge naturally as conal constraints on dual problems. Riccati equations  characterize extrema of these special  LMIs, and therefore provide solutions to the dual problems.  The state-feedback structure of all optimal  signals in these problems emerge out of  alignment (complementary slackness) conditions between primal and dual problems. Perhaps the new insight gained from this approach is that first LMIs,   and then second, Riccati equations arise naturally in  dual, rather than  primal problems.  Furthermore, while traditional LQ problems are set up in $\sfL^2$ spaces of signals, their equivalent covariance-representation problems are most naturally set up in $\sfL^1$ spaces of matrix-valued signals. 
\end{abstract}



\section{Introduction and Motivation}

Linear Quadratic (LQ) control problems in systems and controls  first arose through the original Linear Quadratic Regulator (LQR)~\cite{kalman1960contributions}, which is an optimal control problem, as well as the celebrated Kalman-Yacubovic-Popov (KYP) Lemma~\cite{popov1964hyperstability,kalman1963lyapunov,yakubovich1962solution}. The KYP Lemma can be considered  as a test for an Integral Quadratic Constraint (IQC), which  can be phrased as whether an LQ optimal control problem has finite or infinite infima as advocated in the influential paper of Willems~\cite{willems1971least}. 
Other IQC tests can be used to characterize robust stability of feedback systems subject to uncertainties that can be characterized by IQCs~\cite{megretski1997system}. Those include the Bounded Real Lemma for testing a system's 
$\sfH^\infty$ ($\sfL^2$-induced) norm, as well as the Positive Real Lemma for testing a system's passivity. 
In the same manner as~\cite{willems1971least},  by LQ problems we mean something more general than the LQR problem, 
namely any problem involving linear dynamics with inputs, and a quadratic form defined jointly on the state and input. The goal is to characterize the extrema of the quadratic form  subject to the dynamics as a constraint. 
The literature on these problems is 
vast, and will not be summarized here. Notably, Linear Matrix Inequalities (LMIs) and Riccati equations appear frequently as 
central characters in these intertwined stories. 

Connections between LQ problems and LMIs were pointed out by Willems~\cite{willems1971least}. The books~\cite{boyd1994linear,boyd2004convex} (see also~\cite{scherer2000linear}) have since popularized  the many uses of LMIs in systems and controls, as well as more general optimization problems. Another theme in~\cite{boyd1994linear,boyd2004convex}  is that once a controls problem is formulated as an LMI, then its solution via semidefinite programming is readily achieved. The literature on LMIs is also vast, including a great variety of analysis and synthesis methods, and cannot be summarized here. It is however hard to escape the impression that many arguments with LMIs involve what might be called ``algebraic acrobatics''; an LMI for a particular problem is proposed by a skillful acrobat, after which the verification of whether this LMI characterizes the problem at hand is given.  One  goal of this paper is to step back to try to see a more natural way in which LMIs arise in LQ problems. 

It is well known that quadratic programs can be converted to semidefinite programs using Schur complements~\cite{boyd2004convex}. Alternatively, for problems with purely quadratic objectives, a reparameterization in terms  of the covariance matrix  of the variables (tensor product of vectors) renders the objective linear. This idea appears in optimal control problems in~\cite{gattami2009generalized,gattami2007optimal}, where a discrete-time stochastic LQ problem is reformulated such that the joint covariance of the state and control is the new state. The dynamics of covariances are linear but underdetermined,  the objective becomes linear, and  there are additional positivity constraints on the state expressing that covariance matrices must be non-negative definite. Such a finite-horizon problem then becomes a finite semidefinite program, and it is then shown that the dual problem naturally leads to an LMI. With a similar approach~\cite{you2014h,you2015primal,gattami2015simple} the covariance representation has been used for $\sfH^\infty$ analysis. More recently~\cite{rantzer2023linear}, a joint ``empirical covariance'' of the state and control vectors is used as the ``state'' of a data-driven controller, indicating perhaps that covariance representations are suitable for data-driven control methods as well.

 In this paper, a variation of the covariance representation method of~\cite{gattami2009generalized} is used with a slightly different treatment of the dual problem.  Both stochastic and deterministic problems are considered, where the ``deterministic covariance'' is the rank-one matrix of the outer product of the joint state/input vector. We address continuous-time finite-horizon LQ problems, and view this covariance representation together with its conal constraints as an infinite-dimensional semidefinite program. A natural approach to such problems is to use linear-conic duality and investigate the dual problem. This is where  Differential Linear Matrix Inequalities (DLMIs) show up naturally as  dual constraints. Since the dual objective is also linear, the dual problem is then solved by extremal solutions (in the Loewner ordering on matrices) of the DLMIs, which in turn are given by Differential Riccati Equations (DREs). We use a Banach space version of weak duality that has a simple proof. Rather than give technical conditions for strong duality, we instead use a complementary slackness (alignment) condition between primal and dual problems. Alignment  actually gives additional insight into the solutions, as it  shows that all optimal input signals for LQ problems are of the form of static state feedback. 

The use of duality in optimal control has some history. In our present context, an LQ problem in  covariance representation is most naturally set up in a matrix-valued $\sfL^1[0,\sT)$ signal space. Other examples of $\sfL^1$ optimization in control include~\cite{dahleh19861,dahleh19871} where  duality was used to characterize optimal closed loops as ``sparse'' impulse responses. Later on~\cite{qi2001matlab,salapaka1999multi} finite-dimensional approximations to infinite-dimensional dual problems were used to obtain convergent confidence intervals for numerical solutions to mixed-norm robust control problems in a similar spirit to primal-dual algorithms in semidefinite programming~\cite{boyd2004convex,vandenberghe1996semidefinite}. 

In a general setting, Vinter~\cite[Theorem 2.1]{vinter1993convex} developed a duality framework for nonlinear  optimal control problems in which the dual constraint is a Partial Differential Inequality (PDI). The number of independent variables in the PDI is one (for time) plus the number of states. In our present context, the dual constraint of an LQ problem is a DLMI, i.e. a ordinary differential inequality. The relation between the PDI of~\cite{vinter1993convex} and the DLMI presented here appears similar  to that between the HJB equation of dynamic programming (a PDE) whose solutions for LQ problems are quadratic functionals parameterized by solutions to  matrix-valued differential Riccati equations.  It appears that when restricted to LQ problems, the PDI of~\cite{vinter1993convex} collapses to the DLMIs presented here, although those exact connections remain to be explored. 

In~\cite{rantzer2001dual}, Rantzer introduced a criterion for stability from almost all initial conditions in terms of a PDI (without time dependence). Density functionals that satisfy this PDI are duals of classic Lyapunov functions. The interpretation of a Lyapunov functions as a cost-to-go in an optimal control problem gives a similar interpretation of a density functional in a dual optimal control problem. More recently~\cite{mauroy2020koopman}, duality between this type of analysis and the Koopman representation of dynamical systems have been explored. We again point out that in the present context of LQ problems,  DLMIs are parameterized by matrices, which are finite dimensional objects in contrast to the infinite-dimensional objects which solve the PDIs of the more general criteria. 

Positive dynamical systems are those whose states have positive components, or more generally  evolve in positive cones~\cite{angeli2003monotone,rantzer2018tutorial}. 
For optimal control of positive systems with linear objectives, it is natural to use a dual formulation of the problem as was recently done to provide explicit solutions to the associated HJB equations~\cite{rantzer2022explicit}. For any linear dynamical system, the covariance representation is a linear positive system (in the Loewner order of positive semi-definite matrices), and thus ideas from positive systems appear more generally applicable to not-necessarily-positive systems. A similar theme is used in~\cite{bamieh2020input} for analysis of arbitrary linear systems with multiplicative stochastic uncertainty, where an equivalent deterministic system acting on matrix-valued covariance signals is a monotone system.

We finally mention that some of the statements with DLMIs in this paper resemble those in~\cite{seiler2014stability} where finite time-horizon IQC problems are investigated, though the techniques we use are somewhat different. In fact, essentially all the results in this paper have appeared elsewhere, and most are a reworking of the ideas of~\cite{willems1971least} using duality.  Thus the novelty is not in the set of ideas presented, but perhaps in the {\em sequencing} of these ideas. Whether this is compelling or not is probably a matter of taste. The next subsections summarize the problem formulation and the sequence of steps for using duality to arrive at solutions in terms of DLMIs and DREs. The remainder of the paper is devoted to the details and the necessary background. The paper is written in a tutorial style with an attempt at self containment. Thus many facts that appear in the literature are repeated here for that purpose.

\subsubsection*{Notation and Terminology} 

We use the term ``positive matrix'' to refer to a positive semi-definite matrix, and say ``strictly positive'' to refer to a positive-definite matrix. The notation $N^*$ is used for the transpose of a matrix $N$. Matrices are generally (though not always) denoted by capital letters such as $A$, while operators are generally denoted with calligraphic script, e.g. $\cL(X):= AX+XA^*$. Capital sans-serif font (such as $\sfV$ or $\sfL^1$), or ``blackboard bold'' (such as $\R$ or $\bbP$)  are used for sets and vector spaces.

\subsection{Problem Formulation}

 Linear-Quadratic Problems (LQPs) in control systems are those where the dynamics are linear, and 
a cost function is a quadratic form on all the signals in the system 
\be
\begin{aligned} 
	\xd ~&=~  Ax +  B v , 
				&&  
				\arraycolsep=2pt
				\begin{array}{rcl} x(t_\rmi) &=&  \rmx_\rmi	,\\ x(t_\rmf) & = & \rmx_\rmf ,	\end{array} 			
																			 	\\ 
	 \qform(x,v)  ~&=~ \int_{t_\rmi}^{t_\rmf} \bbm x \\ v \ebm^* \bigmat{\calQ}  \bbm x \\ v \ebm dt 
		~=:~ \int_{t_\rmi}^{t_\rmf} \sfq(x,v) ~dt , 													
\end{aligned} 
\label{LQP.eq}
\ee
where either initial  $\rmx_\rmi$ or final $\rmx_\rmf$ state maybe specified or free, and initial and final times $t_\rmi$ and 
$t_\rmf$ may be finite or infinite.  
$\calQ$ is a matrix that determines the quadratic form $\qform(x,v)$. 
This is a powerful general framework  which encompasses many different problems. 
\begin{itemize} 
	\item 
		The signal $v$ can play different roles. In the LQR problem it is a control, and the problem is to determine the 
		optimal control that minimizes $\qform$. In other problems such as $\sfH^\infty$ analysis or more generally 
		 Integral Quadratic 
		Constraints (IQCs), $v$ is an exogenous signal and the analysis problem is to determine whether $\qform$ 
		remains positive or finite for all possible such exogenous signals. 
	\item 
		Without loss of generality, the matrix $\calQ$ is taken as symmetric. 
		It may be positive or have mixed signature depending on the problem at hand. 
	\item The time horizon $t_\rmf-t_\rmi$ can be finite or infinite. Usually for finite horizons, there are no conceptual differences 
		between the cases of time-invariant or time-varying systems. For  infinite horizon problems issues of 
		stability arise that need to be characterized. 
\end{itemize} 

The diversity of the problems that can be formulated as~\req{LQP} is best appreciated by listing some of the 
well-known ones. The following is a non-comprehensive list by way of examples. 
\begin{itemize} 
	\item \textsf{The Linear Quadratic Regulator (LQR):} In this problem the input signal $v$ is a control, 
		traditionally denoted by $u$.  Given an initial state $\rmx_\rmi$, it is desired 
		to drive (``regulate'') this state to zero while minimizing an objective that is a combination of regulation and 
		control costs 
		\[
			\sfq(x,u) ~:=~ x^*Qx + 2x^*Nu + u^* R u ~=~ \bbm x \\ u \ebm^* \bbm Q & N \\ N^* & R \ebm \bbm x \\ u \ebm  
			~=:~  \bbm x \\ u \ebm^* \bigmat{\calQ}  \bbm x \\ u \ebm
		\]
		In this problem $x(t_\rmi)=\rmx_\rmi$ is given and $x(t_\rmf)$ is free. The
		``disturbance'' in this setting is the non-zero initial state, and the task of the control $u$ is 
		to ameliorate the effect of this initial disturbance by regulating the state back to zero (equilibrium). 
		In this setting, $x$ is uniquely determined by $u$, and the 
		 task is to find 
		$u$ that minimizes this objective
		\[
			\inf_u \qform(x,u) , 
			\hstm 
			\mbox{subject to}
			\hstm 
			\xd=Ax+Bu, \hsom x(t_\rmi)=\rmx_\rmi. 
		\]
		
		In this problem, we assume $Q\geq 0$, $R>0$, and $N$ is chosen subject to the constraint 
		 $\calQ\geq 0$ (not to be confused with its submatrix $Q$). 
		
	\item 
		\textsf{Kalman-Yacubovic-Popov (KYP) Lemma and Integral Quadratic Constraints (IQCs):} 
		The KYP Lemma is fundamentally about whether a quadratic form like~\req{LQP} is nonnegative (or alternatively, 
		nonpositive) for all signals consistent with the system dynamics. It is usually stated for infinite time-horizon 
		problems, and has frequency-domain as well as a time-domain characterizations. 
		More modern uses of the lemma use the formalism of IQCs for various robustness
		analysis problems. 
		The two most commonly used
		instances of the KYP lemma are the ``positive real'', and the ``bounded real'' lemmas. 
		
		\begin{itemize} 
		\item 
            		\textsf{Characterizing Passivity:} (aka The Positive-Real Lemma) 
			 This is an analysis problem, so the input is viewed as an exogenous signal in $\sfL^2[t_\rmi,t_\rmf]$.  
            		A system is called passive from $v$ to an output $z=Cx+Dv$ if the inner product between the input and output
            		 is positive. In this case 
            		\be
            			\sfq(x,v) ~:=~z^*v = (Cx+Dv)^*v = x^*C^*v+ v^*D^* v = 
            			\arraycolsep=1.5pt
            			\frac{1}{2} \bbm x \\ v \ebm^* \bbm 0 & C^* \\ C & D\text{+} D^* \ebm \bbm x \\ v \ebm, 
            		  \label{Passivity_IQC_set.eq}		  
            		\ee
            		and  initial conditions are zero $x(0)=0$. 
            		In this problem  we want to check whether this quadratic
            		form is positive for all signals $(x,v)$ consistent with the system's equations  
            		\[
            			\forall w, ~\smint{t_\rmi}{t_\rmf} \sfq(x,v) ~dt ~\geq~ 0
            			\hstm \Leftrightarrow \hstm 
            			\inf_w \smint{t_\rmi}{t_\rmf} \sfq(x,v) ~dt ~=~0. 
            		\]
		\item 
            		\textsf{Characterizing $\sfL^2$-induced ($\Hinfty$) Norms:} (aka The Bounded-Real Lemma): 
			This is another analysis problem with zero initial state 
            		$x(0)=0$.  The system has $\Ltwo(0,\sT)$-induced norm less than a given number $\gamma$ from 
            		 $v$ to an output $z=Cx$  iff $\int_{0}^{\ssT} 	\sfq(x,v) dt \geq 0$, where 
            		\be
            			\sfq(x,w) 
            			~:=~		\gamma^2 v^*v - z^*z  ~=~ \gamma^2 v^*v - x^*C^*C x ~=~
            			\arraycolsep=1.5pt
            			 \bbm x \\ v \ebm^* \bbm \sm C^*C & 0 \\ 0 & \gamma^2 I  \ebm \bbm x \\ v \ebm, 
            		  \label{Hinf_IQC_set.eq}
            		\ee
            		for all  inputs $v$. Again, this is a problem of characterizing whether  a quadratic form defined 
            		on the system's signals is always positive, which as in the passivity problem, is the same as 
            		insuring that the infimum over all $v$ is zero. 
            		
            		Note that an  equivalent formulation would also be 
            		\[
            			\sfq(x,v) ~:=~ x^*C^*C x - \gamma^2 v^*v  ~=~
            			\arraycolsep=1.5pt
            			 \bbm x \\ v \ebm^{{*}}  \bbm  C^*C & 0 \\ 0 & \sm\gamma^2 I  \ebm \bbm x \\ v \ebm , 
            		\]
            		with the condition that $\int_{0}^{\ssT} 	\sfq(x,v) dt \leq 0$, i.e. checking that the supremum of this 
            		quadratic form is zero.  This alternative, but equivalent, formulation 
            		is sometimes encountered in the literature. We adopt here the formulation~\req{Hinf_IQC_set} instead since 
            		we want to treat all problems as infimization problems so that a single result can be stated for all LQP problems 
            		in a manner similar to  the LQR problem. 
		\end{itemize} 

%
%

\end{itemize}

Derivations of the statements above regarding the $\Hinfty$ norm and passivity are included in Appendix~\ref{HinfPass.appen} 
for reference. 
As already stated, we  adopt a unifying convention for all problems 
whereby  $R>0$ with no restriction on the remaining matrices.


\subsection{Outline of the Present Approach}  								\label{outline.sec}
	
	As an overview, we now summarize  the main steps of the present approach, with details to be presented in later 
	sections. 
	The key idea  is to recast the linear-quadratic problem~\req{LQP} as a linear-linear problem
	for {\em matrix-valued signals}, but with an additional positivity constraint as follows. 
	Restating~\req{LQP} 
            \begin{align} 
            	\xd ~&=~   Ax  +  B v	,		
				\hstm 	x(0)=\rmx_\rmi,	
									 			            \label{LQP_two.eq}				\\ 
            	\qform(x,w) ~&=~ \int_{0}^{\ssT}  \bbm x \\ v \ebm^* \bigmat{\calQ} \bbm x \\ v \ebm dt 
            		~=:~ \int_{0}^{\ssT} \sfq(x,v) ~dt , 												
            \end{align} 
	observe that if we define the following {\em matrix-valued} signal, which is the outer (tensor) product of the 
	original problem variables 
	\be
		\Sigma(t) ~:=~ \bbm x(t) \\ v(t) \ebm \bbm x^*(t) & v^*(t) \ebm 
				~=:~ \bbm \Sigma_{xx}(t) & \Sigma_{xv}(t) \\ \Sigma_{xv}^*(t) & \Sigma_{vv}(t) \ebm ,
		\hstm 
	 \label{Cov_LQP.eq}
	\ee
	then the quadratic objective $\qform$ becomes {\em linear} in $\Sigma$ 
	\be
		\qform ~:= \!\!
			 \int_{0}^{\ssT}    \!\!{ \bbm x^* & v^* \ebm \bigmat{\calQ}  \bbm x \\ v \ebm } dt
		      =
			 \int_{0}^{\ssT}  \!\! \trcc{   \bigmat{\calQ}  { \bbm x \\ v \ebm  \bbm x^* & v^* \ebm }} dt
		 	=: \int_{0}^{\ssT}  \inprod{\calQ}{\Sigma	\romn} ~dt , 
	  \label{Q_Sigma_form.eq}
	\ee
	where $\inprod{\calQ}{\Sigma} := \trcc{\calQ^*\Sigma}$ is the inner product on 
	matrices\footnote{Note that for any two matrices $M_1$ and $M_2$ of compatible dimensions, 
	 	$\trcc{M_1M_2}=\trcc{M_2M_1}$. Thus for any number of matrices, one can perform ``cyclic permutations'' 
		inside the trace $\trcc{M_1 M_2 \cdots M_n} = \trcc{ M_2 \cdots M_{n} M_1}$. This fact was used 
		in the second equality in~\req{Q_Sigma_form}.  }. 
	If $v$ is stochastic, then we would take expectations in~\req{Cov_LQP}, and 
	 $\Sigma$ would then be  the joint input-state covariance matrix. If $v$ is deterministic, 
	then  $\Sigma$ can still be thought of as a ``deterministic covariance'', or as already stated,  the outer (tensor) 
	product of the vector signal $(x,v)$ with itself. In this case, $\Sigma$ has the additional property 
	of always being a rank-one matrix. 
	
	A differential equation  for the {\em $\Sigma_{xx}$ portion} of $\Sigma$ can be derived from 
	the system dynamics~\req{LQP_two} as follows
	\begin{align}
		\frac{d}{dt} xx^* ~&=~ \xd x^* + x \xd^* ~=~ (Ax+Bv) ~x^* + x~(Ax+Bv)^*			\nonumber		\\
		 &=~ A~xx^*+B~vx^* + xx^*~A^*+xv^*~B^* 	
																			\nonumber		\\
		\Leftrightarrow ~~
		\dot{\Sigma}_{xx}
		& = \bbm A & B \ebm \bbm \Sigma_{xx} \\ \Sigma_{xv}^* \ebm  + 
				\bbm \Sigma_{xx} & \Sigma_{xv} \ebm \bbm A^* \\ B^* \ebm 		
							& 		\Sigma_{xx}(0) &= \rmx_\rmi\rmx_\rmi^*	
																			\nonumber		\\
		\hspace{-2em}
		\Leftrightarrow ~~ 
		\Em 
			\bigmat{\tcdb{\dot{\Sigma}}} 
		\Ems
		&= \bbm A & B \ebm 
			\bigmat{\tcdb{\Sigma}}   
			\Ems + \Em 
			\bigmat{\tcdb{\Sigma}} 
			\bbm A^* \\ B^* \ebm  				, 
						&					\Em \bigmat{\!\! \tcdb{\Sigma(0)} \!\!\!\!} \Ems &= \rmx_\rmi\rmx_\rmi^*,
																\label{Sigma_diffeq_LQP.eq}
	\end{align} 
	where matrix dimensions are shown for emphasis. 
	This is a linear, but somewhat unusual   differential equation. 
	Only a subset of the components of the derivative $\Sigmad(t)$ are given as a linear function of $\Sigma(t)$. 
	Thus this is a highly under-deteremined system of differential equations, and has non-unique solutions 
	even if full  initial or final conditions are specified. 
	
	To appreciate the structure of~\req{Sigma_diffeq_LQP} more precisely, define the 
	linear, matrix-valued operators 
	on matrices 
	 \be
	 	\cE \big( \Sigma \big) := \Em \bigmat{\Sigma} \Ems , 
		\hstm 
		\cA \big(\Sigma  \big) := \bbm A & B \ebm \bigmat{\Sigma} \Ems
			+ \Em  \bigmat{\Sigma}   \bbm A^* \\ B^* \ebm, 
	   \label{mat_op_LQP.eq}
	 \ee
	 and the differential equation~\req{Sigma_diffeq_LQP}  can now be written as 
	 \be
	 	\cE\big(\Sigmad(t)\big) ~=~ \cA\big( \Sigma(t) \big), 
		\hstm\hstm 
		\cE\big( \Sigma(0) \big) = \rmx_\rmi \rmx_\rmi^* =: \sfX_\rmi . 
	   \label{Sigma_diff_sum.eq}
	 \ee
	 Note that the operator $\cE$ is not invertible, it takes $(n+m)\times(n+m)$ matrices (where $n:=\dim(x)$ and $m:=\dim(v)$) 
	 to $n\times n$ matrices. The differential equation above 
	  thus belongs to the class of ``descriptor systems'', and  non-uniqueness of solutions is  typical for some 
	  classes of descriptor systems. 
	 It is important to understand the relation between the original vector differential equation~\req{LQP_two} and the 
	 matrix differential equation~\req{Sigma_diffeq_LQP}.
	 In the deterministic case, 
	 any solution of the 
	  vector differential equation gives a solution of the matrix differential 
	 equation by taking the outer 
	 product of the solution vector $(x,v)$ with itself. However, to go from solutions of the matrix to the original vector equation, 
	 one must impose a rank-one constraint on $\Sigma$  as follows
	 \be
	\begin{aligned}
		\cE\big( \Sigmad \big) &= \cA\big( \Sigma \big)  , && t\in[0,\sT]	,		\\
		\hstm \Sigma &\geq 0, ~\rank{\Sigma}=1, ~\cE\big( \Sigma(0) \big) = \rmx_\rmi\rmx_\rmi^*, 
		\label{Sigma_diff_lemma.eq}		
	\end{aligned} 
	\ee
	where $\rank{\Sigma}:={\rm rank}(\Sigma)$. 
	 Given any solution of~\req{Sigma_diff_lemma} we can factor the rank-one  positive matrix $\Sigma(t) = z(t)z^*(t)$, 
	 and then partition the vector $z$ conformably with the partitions of $\cE$ and $\cA$ and obtain
	 $z(t) =: \bbm x(t) \\ v(t) \ebm$ which solve the original differential equation~\req{LQP_two}. Thus there is a one-to-one
	 correspondence between solutions of~\req{LQP_two} and solutions of~\req{Sigma_diff_lemma}. 
	 	
	Now the original LQ problem and the 
	equivalent covariance problem can be stated as follows 
	\be
		\inf_{	\arraycolsep=1pt	
				\begin{array}{rcl} 
					\scpt		\xd	&\scpt	=	&	\scpt		Ax+Bv		\\
					\scpt		\rmx_\rmi & \scpt	=	&	\scpt	 	x(0)
				\end{array}
			} 
			\smint{0}{\ssT} \sfq(x,v) dt  
		\hstm =  
		\inf_{	\arraycolsep=1pt
				\begin{array}{rcl} 
					\scpt		\cE(\Sigmad)  		& \scpt	=	&	 \scpt	 \cA(\Sigma) ,  \\
					\scpt 	\cE(\Sigma(0))	& \scpt	=	&	\scpt		\sfX_\rmi , ~
					\scpt		 \ranks{\Sigma} 		 \scpt	=		\scpt		 1, ~	 \Sigma\geq 0 
				\end{array}
			}			
			\smint{0}{\ssT}  \inprod{\calQ}{\Sigma} dt 
	 \label{LQP_conic_equiv.eq} 
	\ee
	The covariance problem now has linear dynamics, and a {\em linear objective}. We have however two additional 
	constraints, a positivity (cone) constraint $\Sigma\geq 0$, and a rank constraint. The latter makes the problem non-convex
	even if the original linear-quadratic problem is convex. This turns out to be a red-herring. 
	As explained later in~\req{opt_conv_hull}, since the objective is a linear functional, the non-convex rank-one cone of 
	matrices can be replaced by its convex hull, the cone of all positive matrices. In addition, rank-one solutions always 
	exists even if the infimum is taken over the entire cone of positive matrices. 
	

	The important new ingredient in~\req{LQP_conic_equiv} is  the positivity, or {\em cone constraint} $\Sigma\geq0$. 
	Optimization problems with linear constraints and objectives together with a cone constraint  belong to the class of  
	linear-conic optimization problems. The main 
	tool for the study of such problems is {\em linear-conic duality} (summarized in Section~\ref{C_duality.sec}). 
	For the current example, we preview the result stated in Theorem~\ref{IQC_i.thm}, and how it can be used to address 
	all LQ Problems in four main steps. 
	
	\begin{enumerate} 
		\item The differential equation~\req{Sigma_diff_sum} for $\Sigma$ can be converted to a linear constraint in function 
		space by integrating it forward in time 
		\be
			\left. 
			\begin{array}{rcl} 
				\cE\big(\Sigmad\big) &=&  \cA\big( \Sigma \big) 	\\ 
				\cE(\Sigma(0)) &=& \sfX_\rmi 
			\end{array} 
			\right\} 
			\hstm \Leftrightarrow \hstm 
			\lb \cE -\intf \cA \rb(\Sigma) ~=~ \heavi \sfX_\rmi , 
		  \label{abst_eq_const.eq} 
		\ee
		where $\intf$ is the forward integration operator, and $\heavi$ is the unit-step (Heaviside) 
		function (thus $\heavi \sfX_\rmi$
		is a constant, matrix-valued function with value $\sfX_\rmi$). $\lb \cE -\intf \cA \rb$ is a bounded 
		operator on (matrix-valued) $\sfL^1[0,\sT]$, and
		equation~\req{abst_eq_const} is a linear equality constraint on $\Sigma$. 
	
		\item 
                    	The ``weak duality'' statement 
                    	for the covariance problem~\req{LQP_conic_equiv} is 
                    	 \be
                    			{\arraycolsep=2pt
                    			\begin{array}{rcl} 
                    				\displaystyle 
                    				\inf_{\Sigma \in \sfL^1(0,\ssT)}	 & &\displaystyle	 \smint{0}{\ssT}  \inprod{\calQ}{\Sigma} ~dt		\\ 
                    				\big(  \cE - \intf\cA  \big) (\Sigma) &=& \heavi \sfX_\rmi		\rule{0em}{1.5em} \\ 
						\ranks{\Sigma} 		 &	=	&	1, ~	 \Sigma\geq 0
                    			\end{array} 	}	
                     		\hstm\hstm	\geq \hstm
                    			{\arraycolsep=2pt
                    			\begin{array}{rcl} 
                    				\displaystyle 
                    				\sup_{Y\in\sfL^\infty(0,\ssT)}	 & & \displaystyle \smint{0}{T}   \inprod{Y}{\sfX_\rmi} 	~dt	\\ 
                    				\calQ - \big(  \cE -\intf \cA \big)^\adj (Y) &\geq& 0		\rule{0em}{1.5em} 
                    			\end{array}
					} 		
                    	 \label{prim_dual_intro.eq}
			\ee
                    	where the dual variable $Y$ is a matrix-valued function over $[0,\sT]$. 
			Note that since $(x,v)$ are in $\sfL^2(0,\sT)$, then their outer product 
			 $\Sigma \in \sfL^1_{\bbS_n}(0,\sT)$, the $\sfL^1$ 
			space of functions taking values in $\bbS_n$, the space of $n\times n$ symmetric matrices. 

                    	The key to the dual problem is the dual inequality constraint $\calQ - \big(  \cE -\intf \cA \big)^\adj (Y) \geq 0$, 
			whose structure depends on the adjoint operator $\big(  \cE - \intf\cA \big)^\adj$. A simple adjoint calculation 
			(Section~\ref{adjoint_systems.sec}) 
			shows that 
			\begin{align*} 
				\big(  \cE  -\intf \cA \big)^\adj(Y) ~=~ \cE^\adj(Y) - \cA^\adj ( \intb Y ) 
				~&=~ 
				\bbm I \\ 0 \ebm Y \bbm I & 0 \ebm 
				- 
				\bbm A^* \\ B^* \ebm \big( \intb Y \big) \bbm I & 0 \ebm 
				-
				\bbm I \\ 0 \ebm \big( \intb Y \big) \bbm A & B \ebm 							\\
				\Rightarrow \hstm 
				\calQ-(\cE-\intf \cA)^*(Y) 
				~&=~ 
				\calQ - \lb 
				\cE\lb \Lambdad \rb - \cA^*(\Lambda) \rb 
				~=~ 
				\calQ + 
				\bbm 	\Lambdad + A^*\Lambda  + \Lambda A &   \Lambda B \\  
                    							B^*\Lambda &		 0 \ebm, 
			\end{align*} 
			where the dual variable $Y$ is replaced with 
			 $\Lambda := \intb Y$, the backwards integral of $Y$ (which implies $Y=-\Lambdad$). 
			
			Thus a Differential Linear Matrix Inequality (DLMI) arises out of  the dual cone constraint 
			$\big(  \cE -\intf \cA \big)^\adj(Y)\geq 0$ when expressed in terms of the backwards integral of $Y$. 
			Note how the structure of the DLMI is really a corollary of the structure of the operators $\cE$, 
			$\cA$, and their adjoints! 
			
			
			Conic duality thus shows that LMIs arise naturally from problems dual to {\em covariance} 
			optimization problems. A similar statement is much less transparent if one only considers 
			the original problem statement~\req{LQP_two} without a covariance reformulation. 
			
			
		\item 
			Differential 
			Riccati equations give extremal solutions to DLMIs, which are then shown to be the solutions
			to the dual problem since the objective is a  linear functional  in the dual variables $Y$ or 
			$\Lambda$. Thus,  Riccati
			equations also arise naturally from conic duality. However,  In this framework  
			the DLMIs {\em come first}, and Riccati 
			equations arise second as extremal solutions to the DLMIs. Riccati equations can be thought of 
			in this context as a technique to solve these special DLMIs (and LMIs) rather than using general-purpose convex 
			optimization methods. 
		\item 
			The duality gap will be shown to be zero by constructing optimal primal and dual solutions
			together with an {\em alignment} (complementary slackness) condition between them. 
			Specifically, linear-conic duality states that if there exists primal and dual variables $\Sigmab$ and $\Lambdab$ 
			such that 
			\be
				0 ~=~ 
				\inprod{\calQ - \cE^*(\dot{\Lambdab})-\cA^*(\Lambdab)}{\Sigmab} , 
			 \label{alignment_intro.eq}
			\ee
			then $\Sigmab$ and $\Lambdab$ are optimal for the primal and dual respectively, and the duality gap is zero. 
			The optimal solution $\Lambdab$ to the dual problem is the solution of a Riccati equation. A classic argument 
			implies that this solution gives a minimal-rank factorization of the LMI matrix as 
			follows\footnote{This expression is shown here for the notationally simpler case with $N=0$.} 
			\[
				\calQ - \cE^*(\dot{\Lambdab})-\cA^*(\Lambdab)
				= 
				\calQ +  \bbm 	\dot{\Lambdab} + A^*\Lambdab  + \Lambdab A &   \Lambdab B \\  
                    							B^*\Lambdab &		 0 \ebm
				= 	
				\bbm \Lambdab B R^{\sm\frac{1}{2}}  \\  R^{\frac{1}{2}} \ebm  
            			\bbm  R^{\sm\frac{1}{2}} B^* \Lambdab B &  R^{\frac{1}{2}} \ebm . 
			\]
			Since any possible  $\Sigmab$ must be  of the form~\req{Cov_LQP}, the alignment 
			condition~\req{alignment_intro} then reads 
			\begin{align*} 
				0 =  \inprod{\bbm \Lambdab BR^{\sm\frac{1}{2}}  \\ R^{\frac{1}{2}} \ebm  
					\bbm R^{\sm\frac{1}{2}} B^* \Lambdab & R^{\frac{1}{2}}  \ebm }
					{  \bbm \xb \\ \vb \ebm \bbm \xb^* & \vb^* \ebm}
				\hstm &\Rightarrow \hstm 
				\bbm R^{\sm\frac{1}{2}} B^* \Lambdab & R^{\frac{1}{2}}  \ebm 
				 \bbm \xb \\ \vb \ebm	 = 0 												\\
				 &\Rightarrow \hstm 
				\vb ~=~ \lb R^{\sm 1} B^* \Lambdab \rb \xb 	,				 
			\end{align*} 
			where the first implication is a property of symmetric factorizations of any two positive matrices 
			(Lemma~\ref{pos_orth_rank.lemma}). 
			
			Thus the state-feedback nature of  optimal solutions of any LQ problem is a consequence of the alignment 
			condition~\req{alignment_intro} between primal and dual solutions. 
	\end{enumerate}


	This note deals primarily with  finite-horizon cases of LQ problems. The infinite-horizon case will be reported in Part II. 
	The exposition in this note is meant to be tutorial and self contained. 
	The organization is  as follows. Section~\ref{prelim.sec} collects the necessary mathematical preliminaries needed
	to use linear-conic duality in a function (Banach) space, as well as some facts about cones of matrices and their duals. 
	In Section~\ref{DLMIs.sec} we investigate Differential Linear Matrix Inequalities (DLMIs) and show that Differential 
	Riccati Equations (DREs) derived from them give maximal and minimal solutions to the DLMIs. With all the preliminaries 
	established, Section~\ref{LQR_det.sec} 
	gives the solution to the classic Linear Quadratic Regulator (LQR) problem. Section~\ref{LQR_stoch.sec} 
	treats the stochastic 
	version of the LQR problem and illustrates the covariance approach for stochastic problems. 
	 Section~\ref{IQC.sec} treats the general 
	Integral Quadratic Constraints (IQCs) case, and we work out the various special IQC-type problems such as the Bounded 
	Real Lemma and   the Positive Real Lemma in their finite-horizon versions. 
	Some arguments  are relegated to appendices in order not to disrupt the exposition of the key ideas.

\section{Mathematical Preliminaries} 								\label{prelim.sec}

In this section we collect some mathematical preliminaries that will be needed. The presentation is meant to be 
self contained and tutorial. 

\subsection{Matrix Signals and their Norms}

We will be concerned with the vector space $\R^{n\times n}$ of $n\times n$ matrices, and in particular  its subspace 
 $\bbS_n\subset\R^{n\times n}$ of symmetric matrices. 
The vector space $\R^{n\times n}$  has a natural inner product given by 
\begin{align}
	\inprod{H}{M}	~:=~	\trcc{H^* M}. 
  \label{m_in_prod.eq}
\end{align}
This is simply the total sum of the element-by-element product of two matrices $H$ and $M$. Equivalently, if each matrix $H$ and $M$ is expanded as a single vector by stacking its columns over each other, the above sum would simply be the dot product of the two vectors. We note that this inner product generates the Frobenius norm  $\|M\|_F^2 = \trcc{M^*M} = \inprod{M}{M}$  on matrices. 

If we restrict attention to the subspace  $\bbS_n\subset\R^{n\times n}$ of $n\times n$ symmetric matrices, 
the  algebraic dual of $\bbS_n$ is itself since 
 the action of any matrix $H$ on a symmetric matrix $M$ is 
 \[
 	\inprod{H}{M} ~=~ \trcc{H^*M} 
	~=~ \lb \trcc{H^*M}+\trcc{MH} \romn \rb /2 
	~=~ \lb \trcc{H^*M}+\trcc{HM} \romn \rb /2 
	~=~  \trcc{{\scriptstyle \frac{1}{2}} (H^*+H)~M}, 
 \]
 and therefore the action of $\inprod{H}{.}$ on $\bbS_n$ is the same as the action of its symmetrization $\inprod{(H^*+H)/2}{.}$.
 

We will need norms other than the Frobenius norm on matrices, and therefore a different interpretation of~\req{m_in_prod} as a linear functional action (of $H$ on $M$). When $M$ is a symmetric positive matrix (like $\Sigma$ in~\req{Cov_LQP}), its trace is the quantity of interest. The trace on positive matrices is a special case of the ``nuclear norm'' on $n\times n$ matrices defined as 
\be
	\|M\|_1 ~:=~ \sum_{i=1}^n \sigma_i(M) , 
  \label{H1.eq}
\ee
where $\lcb \sigma_i(M)\rcb$ are the singular values of $M$. When $M$ is symmetric and positive, $\trcc{M} = \|M\|_1$, but the definition~\req{H1} applies to any matrix.  
We denote the space of symmetric $n\times n$ matrices with the $\|.\|_1$ norm by $\lb \bbS_n,\|.\|_1 \rb$. 
The norm dual to $\|.\|_1$ is the maximum singular value norm 
\[
	\|H \|_\infty ~:=~ \max_{1\leq i\leq n} \sigma_i(H). 
\]
The fact that this norm is dual to $\|.\|_1$ follows from the following tight ``trace duality'' inequality whose proof is given in Appendix~\ref{trace_duality.sec} 
\[
	\inprod{H}{M} ~:=~ 
	\trcc{HM} 
	~\leq~ \|H\|_\infty \|M\|_1 , 
	\hstm \hstm 
	\sup_{\|M\|_1 \leq 1} \trcc{HM} = \|H\|_\infty .
\]
Thus the space dual to $\lb \bbS_n,\|.\|_1 \rb$ is $\lb \bbS_n,\|.\|_\infty \rb$.
More generally, we can define the space $\lb \bbS_n,\|.\|_p\rb$ for $p\in(1,\infty)$ with the norm 
\[
	\|M\|_p^p ~:=~ \sum_{i=1}^n \lb \sigma_i(M) \rb^p , 
\]
and the space dual to it is $\lb \bbS_n, \|.\|_q \rb$ with $1/p+1/q=1$. However, this level of generality is not needed 
here\footnote{The $\|.\|_p$ norms defined above are the so-called Schatten norms}. 

We will also need to consider norms and linear functional actions on matrix-valued functions of time. 
To motivate the definition we will shortly introduce, consider again the covariance matrix~\req{Cov_LQP} and observe that 
\[
	\int_0^T \trcc{\Sigma(t) \rom} ~dt 
	~=~  \int_0^T \trcc{ \bbm x(t) \\ v(t) \ebm \bbm x^*(t) & v^*(t) \ebm  } ~dt 
	~=~  \int_0^T \lb  x^*(t)x(t)  + v^*(t)v(t) \rom  \rb ~dt , 
\]
where $\sT=\infty$ is a possible limit of this integral. If all signals $(v,x)$ 
need to have finite $\LTwo$ norms, then we should require $\Sigma$ itself to have finite trace norm integrated over the time interval. We therefore require our matrix-valued signals to be in the following Banach space 
 \be
 	\LOne_{\bbS_n} (0,\sT) ~:=~ \lcb M:(0,\sT)\rightarrow \bbS_n ; ~
			\| M \|_1 :=   \smint{0}{T} \| M(t) \|_1 ~dt~< \infty \rcb , 
	\hstm 
	\sT<\infty , ~\mbox{or}~ \sT=\infty. 
  \label{L1def.eq}
 \ee
 Note that the matrix norm $\|M(t)\|_1$ is the one defined in~\req{H1}, and we use the same notation $\|.\|_1$ to denote a matrix norm as well as the above norm on matrix-valued {\em functions}.

 The space dual to~\req{L1def} is the following
 $\LInf$ space
  \be
 	\LInf_{\bbS_n}(0,\sT) :=\lcb H:(0, \sT)\rightarrow \bbS_n ; ~
			\| H \|_\infty := \sup_{t\in(0,\sT)} \lb    \sigma_{\max} \lb H(t) \romn\rb \rom \rb~< \infty \rcb , 
	\hstm 
	\sT<\infty , ~\mbox{or}~ \sT=\infty. 
  \label{Linfdef.eq}
 \ee
 The linear functional action is defined as follows 
 \be
 	\inprod{H}{M} 
	~:=~ \smint{0}{T} \inprod{ H(t) }{M(t) } ~dt 
	~:=~ \smint{0}{T}  \trcc{ \rom H(t) ~M(t) } ~dt , 
	\hstm 
	M\in \LOne (0,\sT), \hstm 
	H\in \LInf (0,\sT)  , 
\ee
where from now on we abbreviate $\LOne_{\bbS_n}$ and $\LInf_{\bbS_n}$ to simply $\LOne$ and $\LInf$. 

The fact that  $\LOne(0,\sT)$ and $\LInf(0,\sT)$ are duals, and the tightness of the inequality 
\[
	\inprod{H}{M} 
	~\leq~ \|H\|_\infty \|M\|_1 , 
	\hstm \hstm 
	\sup_{\|M\|_1 \leq 1} \inprod{H}{M} = \|H\|_\infty ,
\]
follows from the material in Appendix~\ref{trace_duality.sec} and standard functional analytic arguments. 
With this setting, we see that the quadratic form~\req{Q_Sigma_form} should be interpreted as the action of matrix-valued 
function $\calQ$ in $\LInf(0,\sT)$ on the matrix-valued function $\Sigma$ in $\LOne(0,\sT)$.

\subsubsection*{Integration Operators}

In this section we restrict attention to the finite-horizon case $\sT<\infty$. 
Spaces like $\LOne_{\bbS_n}(0,\sT)$ and $\LInf_{\bbS_n}(0,\sT)$ will sometimes be abbreviated as 
$\LOne$ and $\LInf$ for notational simplicity. 

Define the 
  \deffont{forward integration operator} ${\intf}$ 
and its adjoint, the \deffont{backwards integration operator} ${\intb}$
\[
	\big( \intf{X}\big)(t)  ~:=~ \smint{0}{t} X(\tau)~d\tau, 
	\hstm \hstm 
	\big( \intb X \big) (t)  ~:=~ \smint{t}{T} X(\tau) ~d\tau , 
	\hstm\hstm 
	t\in[0,\sT]. 
\]
That these two operators are (formal) adjoints is shown by the following calculation 
\begin{align} 
	\textstyle
	\inprod{Y}{\intf X} 
	~&= \textstyle
		\trcc{	\int_{0}^{T}  Y^*(t) \lb  \int_{0}^t X(\tau) ~d\tau\rb  dt }
	=\trcc{ \int_{0}^{T}	 Y^*(t) \lb \int_{0}^{T} X(\tau)~ \heavi(t\sm\tau)  ~d\tau\rb ~dt 	}		\nonumber	\\
	&= \textstyle
		\trcc{ \int_{0}^{T} \lb \int_{0}^{T} Y^*(t) ~\heavi(t\sm\tau) ~dt \rb    X(\tau)  ~d\tau }
	=	\trcc{ \int_{0}^{T} \lb \int_\tau^{T} Y^*(t) ~dt \rb    X(\tau)  ~d\tau 	}					\nonumber	\\
	&= \textstyle
		\trcc{	\int_{0}^{T} \big( (\intb Y)(\tau) \big)^*    X(\tau)  ~d\tau }
	=~ \inprod{\intb Y }{X} , 															\label{If_Ib_adj.eq}
\end{align} 
where $\heavi(.)$ is the unit-step (Heaviside) function. 

To use duality, we will need to be careful with the signal spaces and their duals. 
%
Forward integration  is a boundeded operator (with bound $1$)
$\intf:\LOne\rightarrow\sfC$, the closed subspace of continuous 
functions\footnote{In fact, it is almost a tautology that $\intf$ maps $\LOne$ to $\sfAC_0$,  the subspace (of $\sfC$)
	 of absolutely continuous functions with zero initial value. It's an isometric isomorphism if one endows 
	 $\sfAC_0$ with the norm pushed forward from $\LOne$ by this map. The inequality~\req{intf_ineq} 
	 shows that with this norm, $\sfAC$ (with this norm)
	  is continuously embedded in $\sfC$ (with the supremum norm). 
	  However, this level of nuance is not needed for the finite-horizon problem.}
in $\LInf$. The bound 
is shown by 
\begin{align} 
					\left\| \lb \intf \Sigma \rb(t) \right\|_\infty
					~&=~\left\|  \smint{0}{t} \Sigma(\tau) ~d\tau \right\|_\infty
					~\leq~  \smint{0}{t} \left\| \Sigma(\tau) \right\|_\infty ~d\tau 
					~\leq~  \smint{0}{t} \left\| \Sigma(\tau) \right\|_1 ~d\tau 
					~\leq~  \smint{0}{T} \left\| \Sigma(\tau) \right\|_1 ~d\tau 			\nonumber	\\
					\Rightarrow \hstm 
					\left\|  \intf \Sigma  \right\|_\infty
					~&=~ 					
						\sup_{t\in(0,T)}	\left\| \lb \intf \Sigma \rb(t) \right\|_\infty
						~\leq~ \|\Sigma\|_1 . 									\label{intf_ineq.eq}
\end{align} 
For finite-horizon problems, we can also regard $\intf$ as a bounded operator $\intf:\LOne\rightarrow\LOne$ due to the bound
\begin{align} 
					\left\| \lb \intf \Sigma \rb \right\|_1
					~&=~  \smint{0}{\ssT} \|   \smint{0}{t}\Sigma(\tau)  ~d\tau \|_1 ~dt
					~\leq~  \smint{0}{\ssT}  \smint{0}{\ssT} \| \Sigma(\tau) \|_1~ \heavi(t\sm\tau) ~d\tau  ~dt
					~=~  \smint{0}{\ssT} \| \Sigma(\tau) \|_1~ \smint{0}{\ssT}  \heavi(t\sm\tau) ~dt  ~d\tau
																\nonumber			 \\
					~&=~  \smint{0}{\ssT} \| \Sigma(\tau) \|_1~(\sT-\tau) ~d\tau
					~\leq~ \sT \smint{0}{\ssT} \| \Sigma(\tau) \|_1~d\tau
					~=~ \sT ~\left\| \Sigma \right\|_1 	.				\label{If_L1L1.eq}
\end{align} 
Note how this does not work for infinite-horizon problems. This setting however simplifies the statement of duality 
in the sequel. We finally note that the above also implies that 
 backwards integration is a bounded operator $\intb:\LInf\rightarrow\LInf$.

\subsubsection*{Matrix Operators} 

When dealing with Lyapunov equations and LMIs, we repeatedly encounter matrix-valued mappings of matrix arguments with terms
having the form 
\[
	\cL(X) ~:=~ AXB, 
\]
where $A$, $B$ and $X$ are of compatible dimensions. This is a linear operator $\cL:\R^{n\times m} \rightarrow \R^{p\times q}$ between 
finite vector spaces with the $\|.\|_1$ norm, and we can therefore speak of its adjoint. It is a simple exercise to show that the 
adjoint $\cL^\adj: \R^{p\times q} \rightarrow \R^{n\times m}$ is 
\be
	\cL^\adj(Y) ~=~ A^* Y B^* , 
  \label{cL_adj.eq}
\ee
where $A^*$ denotes matrix transpose. 

For a time-varying version of this operator, a
 similar exercise gives the adjoint for the following operator on matrix-valued functions 
 $\cL:\LOne \rightarrow \LOne$ 
\be
	\big( \cL(X) \big)(t) := A(t)~X(t) ~B(t), 
	\hstm 
	\big( \cL^\adj(Y) \big)(t) = A^*(t)~Y(t) ~B^*(t),
	\hstm 
	t\in (0,\sT), 
  \label{cL_t_adj.eq}	
\ee
where $\cL^\adj:\LInf \rightarrow \LInf$, and we require the functions $A(.)$ and $B(.)$ to be uniformly bounded. 
Thus the two matrix operators defined in~\req{mat_op_LQP} have as adjoints 
\be
	\cE^\adj\lbb  Y  \rbb := \Ems  \left[ Y \right]   \Em, 
	\hstm 
	\cA^\adj\lbb Y \rbb := \bbm A^* \\ B^* \ebm  \left[ Y \right]  \Em   + 
							\Ems  \left[ Y \right]  	 \bbm A & B \ebm  , 
  \label{EA_adjoints.eq} 
\ee
where the dependence on $t$ is suppressed, and the matrix dimensions are illustrated to contrast with those
of $\cE$ and $\cA$ in~\req{mat_op_LQP}.

\subsection{Differential and Integral Equations and their Adjoints} 								\label{adjoint_systems.sec}

Consider a differential equation of the form 
\be
	\cE\big(\Xd(t) \big) ~=~ \cA\big(X(t) \big) ~+~ \sfW(t), 
	\hstm\hstm 
	t\in[0,\sT], 
	\hstm 
	\cE\big( X(0) \big)  = \sfX_\rmi, 
  \label{abst_sys_E.eq} 
\ee
where $X$ is a vector or a matrix, or more generally an element of a vector space. $\cE$ and $\cA$ are linear operators on matrices, 
and $\sfW(.)$ and $\sfX_\rmi$ are given forcing function and initial condition respectively. If the operator $\cE$ is non-invertible, 
then this equation belongs to the so-called ``descriptor system'' representation, and will typically have non-unique 
solutions\footnote{Note that if $\cE$ is not invertible, then for example $X(0)$ and $\sfX_\rmi$ may not have the same dimensions.} .  
The invertibility of  $\cE$ 
is irrelevant to the material we present next. 

We will need to rewrite~\req{abst_sys_E} as an abstract linear constraint in a function space. The simplest way to do 
this is to recast it as an integral equation by 
integrating both sides of~\req{abst_sys_E} and using the fundamental theorem of calculus
\begin{align} 
	\int_{0}^{t} \cE \!\lb \dot{X}(\tau) \rb d\tau ~&=~ 
	\int_{0}^{t}	\bbm I & 0 \ebm \dot{X}(\tau) \bbm I \\ 0 \ebm d\tau 
	~=~ 
	\cE\big(X(t)\big) ~-~ \cE\big( X(0) \big) 											\nonumber		\\
	\Rightarrow \hstm 
	\cE\big(X(t)\big) ~-~ \cE\big( X(0) \big) 
	~&=  \smint{0}{t} \cA \big( X(\tau) \big) ~d\tau ~+~ \smint{0}{t} \sfW (\tau) ~d\tau 			 	\nonumber	\\
	\Rightarrow \hstm 
	\cE\big(X(t)\big) ~-  \smint{0}{t} \cA \big( X(\tau) \big) ~d\tau 
	~&=   \smint{0}{t} \sfW (\tau) ~d\tau 	~+~ \sfX_\rmi		 							\hstm\hstm t\in[0,\sT] . 
  \label{abst_int.eq} 
\end{align} 
This equation can be expressed  in operator notation using the forward integration operator $\intf$ 
\be
	\big( \cE  ~-~\intf \cA  \big) (X)  ~=~ \intf ~\sfW ~+~ \heavi \sfX_\rmi , 
  \label{sys_int_op.eq} 
\ee
where we utilized a slight abuse of notation with 
\[
	\big( \cE (X)\big) (t) ~:=~ \cE\big( X(t) \big), 
	\hstm 
	\big( \cA(X) \big)(t) ~:=~ \cA\big(X(t) \big) , 
	\hstm\hstm 
	t\in[0,\sT] .
\]
Equation~\req{sys_int_op} is a  linear constraint on the function $X$. The right hand side is a fixed given 
function, and the left hand side is a linear operator acting on $X$. The adjoint of this linear operator plays an important 
role in the sequel
\be
	\lb \cE  - \intf \cA  \rb^\adj ~=~ \cE^\adj - \cA^\adj  \intf^\adj ~=~ \cE^\adj -    \cA^\adj \intb . 
  \label{int_eq_EAX.eq}
\ee	
Note that while the original operator involved forward integration, the adjoint operator involves backwards integration. 


We will utilize the operator~\req{sys_int_op} when acting on functions $X\in\LOne$. Recall from the 
bound~\req{If_L1L1} that $\intf:\LOne\rightarrow\LOne$, and therefore 
$\lb\cE+\intb\cA\rb:\LOne\rightarrow\LOne$ is a bounded operator. 
Similarly $\lb \cE  - \intf \cA  \rb^\adj:\LInf\rightarrow \LInf$ is a bounded operator. 


\subsection{Positive Matrices and their Orthogonality}

The subset (not a subspace) of positive matrices
$\bbS_n$ is denoted by 
\[
	\bbP_n ~:=~ \lcb M\in\bbS_n; ~M\geq 0 \romn \rcb . 
\]
This set has the structure of a ``cone'', for which the formal definition is as follows. 

\begin{definition} 
Let $\sfV$ be a vector space. A set $\sfP\subset\sfV$ is called a \deffont{cone} if it contains all positive scalings of its 
elements
\[
	\forall \alpha\geq 0, 
	\hstm 
	v\in\sfP
	~ \Rightarrow ~ 
	\alpha v\in\sfP. 
\]
Thus a cone is a collection of one-sided  ``rays'' in $\sfV$. We will assume all cones to be closed and
 ``pointed'', i.e. $(-\sfP)\cap\sfP=0$. 
$\sfP$ is called a \deffont{convex cone} if in addition it is a convex set. 
\end{definition} 

\begin{example} 
	The set of all positive matrices 
	\[
		\bbP_n ~:=~ \lcb M\in\bbS_n; ~M\geq 0 \romn \rcb , 
	\]
	is  clearly   a convex cone in $\bbS_n$. 
	In contrast, the set of all positive, rank-one matrices 
	\[
		\bbP_n^1 ~:=~ \lcb M\in\bbS_n; ~M\geq 0, ~\rank{M}=1  \romn \rcb 
	\]
	is a cone, but not a convex set in $\bbS_n$. In fact, $\bbP_n$ is the convex hull of $\bbP_n^1$. This can be seen 
	from the  dyadic decomposition 
	of a positive matrix, which can be  written as a convex combination of rank-one 
	matrices\footnote{The dyadic decomposition of a positive matrix is  
	the sum of rank-one positive matrices (each being the  outer product of each eigenvector with itself), 
	multiplied by the eigenvalues, all of which are non-negative. Renormalizing each term by the sum of the eigenvalues, 
	the dyadic sum becomes a convex combination of positive, rank-one matrices.}.
	
	Geometrically, the set of all {\em positive-definite} matrices forms the interior of $\bbP_n$. The set of all matrices in 
	$\bbP_n$ with rank $r<n$ lie at the boundary of $\bbP_n$. In particular, 
	$\bbP_n^1$ forms part of the boundary of $\bbP_n$, but it is ``big enough'' that its convex hull is the entirety of the 
	cone $\bbP_n$. 
\end{example}

Recall that any positive matrix  has a symmetric  
factorization\footnote{This is not to be confused with the Cholesky factorization, in which the matrix $U$ is required 
	to be lower (upper) triangular. In  symmetric factorizations, the matrix $U$ is not required to have any special 
	structure.}  
of the form
\be
	M~=~ UU^*.  			
  \label{f_rank_fact.eq}
\ee
It follows that $\Ims{M}=\Ims{UU^*}=\Ims{U}$, i.e. 
 the image space $\Ims{M}$ of $M$ is the same as the image 
space $\Ims{U}$. 

Let 
$M_1=U_1 U_1^*$ and $M_2=U_2U_2$ be two positive 
matrices with respective symmetric  factorizations, 
and suppose they are orthogonal, i.e. $\inprod{M_1}{M_2}=0$. Then
\begin{align*} 
	0 ~&=~ \inprod{M_1}{M_2} ~=~ \trcc{M_1M_2 \romn} 
	~=~  \trcc{U_1 U_1^* U_2 U_2^* \romn} 
	~=~  \trcc{U_2^* U_1 U_1^* U_2 \romn} 				\\
	~&=~ \inprod{U_1^*U_2}{U_1^*U_2} 				
	~=~ \left\| U_1^* U_2 \right\|_F^2 .
\end{align*} 
This last statement means that $U_1^*U_2=0$. To appreciate this better, indicate the dimensions 
\[
	\widemat{U_1^*}  \tallmat{U_2} ~=~ \bbm 0 \ebm . 
\]
Thus the columns of $U_1$ are orthogonal to the columns of $U_2$, and the same is true for the image spaces of $M_1$ 
and $M_2$. This statement also constrains the ranks of $M_1$ and $M_2$. 
\begin{lemma} 															\label{pos_orth_rank.lemma} 
	Let $M_1,M_2\geq 0$ be positive matrices of the same dimension.  
	 If they are orthogonal $\inprod{M_1}{M_2}:=\trcc{M_1M_2}=0$, 
	then their image spaces are orthogonal
	\be
		\Ims{M_1} \perp \Ims{M_2} 
		\hstm \Rightarrow \hstm 
		\rank{M_1} + \rank{M_2} ~\leq~ n . 
	  \label{PM_ortho.eq}
	\ee
	 In particular, if $M_1$ is full rank, then $M_2=0$. 
	 
	 In addition, if $M_1=U_1U_1^*$ and $M_2=U_2U_2^*$ are any symmetric factorizations, then 
	  $U_1^*U_2=0$. 
\end{lemma}

This lemma has an interesting geometric interpretation which 
 should be contrasted with standard vector orthogonality. For a vector $v\in\sfV$ to be zero, it has 
to be orthogonal to {\em all} other vectors
\[
	\forall w\in\sfV ~\inprod{w}{v} = 0 
	\hstm \Leftrightarrow \hstm 
	v=0 .
\]
For a positive matrix $V\geq 0$ however, it suffices to show that if it is orthogonal to a {\em single positive-definite} 
matrix, then it must be zero 
\[
	Q>0~\mbox{and}~ \inprod{Q}{V} = 0 
	\hstm \Rightarrow \hstm 
	V = 0, 
\]
This follows from~\req{PM_ortho}. Since $Q>0$ has full rank, then $V$ must have rank $0$. 

The set of positive-definite matrices forms the {\em interior} of the cone $\bbP_n$. Thus no two non-zero matrices in the interior 
of $\bbP_n$ can be orthogonal to each other, and no non-zero positive matrix can be orthogonal to 
any matrix in the interior of $\bbP_n$. 
{\em If two positive matrices are orthogonal, then they both must lie on the boundary of the cone $\bbP_n$}, i.e. neither can be 
positive definite. 
The reader should visualize the positive orthant in $\R^2$ or $\R^3$ (which are also cones in $\R^n$)  for a 
geometric interpretation of this fact. 

\begin{lemma} 														\label{pos_orth_function.lemma} 
	Let $M\in\LOne(0,\sT)$ and $H\in\LInf(0,\sT)$ (where $\sT$ is possibly $\infty$), and let  
	$0\leq M(t)=U_1(t)U_1^*(t)$ and $0\leq H(t)=U_2(t)U_2^*(t)$ be symmetric factorizations. Then 
	\[
		0=
		\inprod{H}{M} = \smint{0}{T} \trcc{H(t) ~M(t)	\rom} ~dt ~=~ 0 
		\hstm \Rightarrow \hstm 
		U_1^*(t)U_2(t) =0, ~ ~t\in(0,\sT) ~\mbox{a.e.}. 
	\]
\end{lemma} 
\begin{proof} Using the symmetric factorization of say $H(.)$ 
\[
	0 = 
	\inprod{H}{M} = 
	\smint{0}{T} \trcc{U_2(t)U_2^*(t) ~M(t)	\romn} ~dt
	=
	\smint{0}{T} \trcc{U_2^*(t) ~M(t)~ U_2(t) 	\romn} ~dt. 
\]
The integrand is a non-negative function, and therefore if its integral is zero, it must be zero almost everywhere in $[0,\sT]$. 
Furthermore, for each $t$ 
\[
	 \trcc{H(t) ~M(t)	\romn}  = 0  
	 \hstm \Rightarrow \hstm 
	 U^*_1(t) U_2(t) = 0 
\]
by Lemma~\ref{pos_orth_function.lemma}. 	
\end{proof} 

\subsection{Linear-Conic Duality} 											\label{C_duality.sec}

\begin{figure}[t] 
		\centering
		\includegraphics[width=0.9\textwidth]{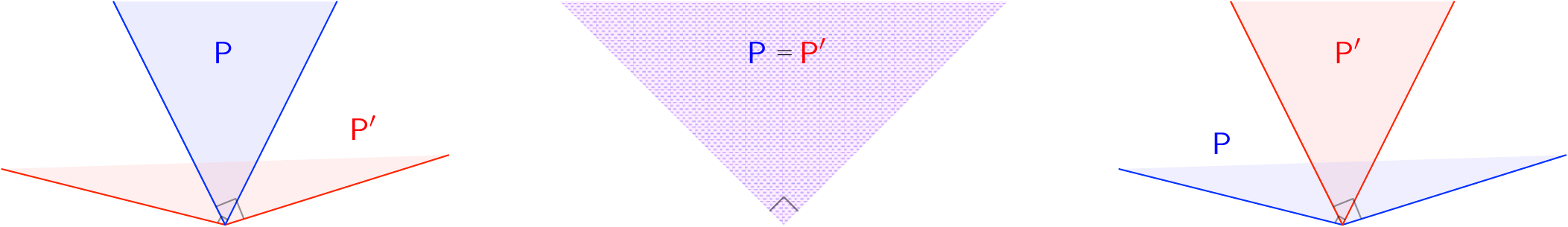}
		
		\mycaption{Examples of several cones $\sfP$ and their duals $\sfP'$ in $\R^2$. Duality is 
			with respect to the standard Euclidean  inner product. If the cone is a ``right-angled wedge'' (center figure), 
			then it is equal 
			to its dual. Otherwise the narrower the cone's angle, the larger is the dual cone's angle as the
			relation~\req{dual_cone_contain} implies. 
			} 
	 	\label{cones_dual.fig} 	
\end{figure} 

The main tool used in this paper is that of {\em linear-conic duality}. To introduce this concept, we first 
define the notion of a dual cone. 
\begin{definition} 
	Let $\sfP\subset \sfV$ be a (not necessarily convex) 
	cone in a  vector space $\sfV$. Its \deffont{dual cone} $\sfP'\subset\sfV'$ is the set of all
	 linear functionals\footnote{If $\sfV$ is a Banach space, then $V'$ is the space of all {\em bounded} linear functionals.} 
	 that are positive on all of $\sfP$ 
	\be
		\sfP' ~:=~ \lcb \romn H\in\sfV' ; ~\forall M\in\sfP, ~ \inprod{H}{M} \geq 0 \rcb .
	  \label{dual_cone.eq}
	\ee
\end{definition} 
Figure~\ref{cones_dual.fig} shows examples of several cones in $\R^2$. Some cones are ``self-dual'', i.e. $\sfP'=\sfP$, but 
most cones in $\R^n$ are not. An intuitive picture to keep in mind is that the smaller the cone is, the larger is its dual 
since the requirement~\req{dual_cone} that elements of $\sfP'$ have to satisfy is less stringent the smaller $\sfP$ is. 
Similarly, the larger a cone is, the smaller is its dual. More precisely, it is easy to show from the definition that 
\be
	\sfP_1 ~\subseteq ~ \sfP_2 
	\hstm \Leftrightarrow \hstm 
	\sfP_2' ~\subseteq ~\sfP_1' . 
  \label{dual_cone_contain.eq}
\ee

\begin{example} 	
	Any positive rank-one matrix can be written as $M=vv^*$ for some vector $v$. Thus another characterization of 
	the set of positive rank-one matrices is  
	\[
		\bbP_n^1 ~=~ \lcb \romn M~=~ vv^*; ~v\in\R^n \rcb . 
	\]
	Thus  a matrix $H$ in the cone dual to $\bbP_n^1$ must satisfy 
	\[
		\forall M\in\bbP_n^1 , ~ 
		\trcc{HM}\geq 0
		\hstm \Leftrightarrow \hstm 
		\forall v\in\R^n , ~ 
		\trcc{H vv^*} = v^*Hv \geq 0 
		\hstm \Leftrightarrow \hstm 
		H\geq 0 , 
	\]
	i.e. $H\in{\bbP_n^1}'$ iff $H\in\bbP_n$. 
	
	Now to calculate the cone dual to $\bbP_n$, recall the symmetric factorization 
	of any positive matrix $M=UU^*$ for some matrix $U$. A 
	 matrix $H$ in the cone dual to $\bbP_n$ then must satisfy 
	\[
		\forall M\in\bbP_n , ~ 
		\trcc{HM}\geq 0
		\hstm \Leftrightarrow \hstm 
		\forall U , ~ 
		\trcc{H UU^*} = \trcc{U^*HU} \geq 0 
		\hstm \Leftrightarrow \hstm 
		H\geq 0 .  
	\]
	We therefore conclude that {\em the dual of either $\bbP_n^1$ or $\bbP_n$ is $\bbP_n$}. In particular the set of all positive matrices 
	is its own dual. 
\end{example}

	An interesting observation from the previous calculation is that the dual of the {\em non-convex} cone $\bbP_n^1$  is the 
	{\em convex cone} $\bbP_n$. This turns out to be true in general. 
	\begin{lemma} 
		Let $\sfP\subset\sfV$ be a (not-necessarily convex) cone. Its dual $\sfP'$ is a convex cone in $\sfV'$. 
		Furthermore, if $\sfV$ 	
		is finite dimensional, then the ``dual of the dual'' is the convex hull of the original cone  
		\[
			\lb \sfP'\rb' ~=~ \Conv{\sfP} 
		\]
	\end{lemma}
	The proof of the first statement is immediate, and is left as an exercise. A good intuition for the second statement is  
	provided by observing that $\lb \sfP'\rb'$ must contain $\sfP$, but also must be convex (since it's a dual), and therefore 
	must at least contain the convex hull of $\sfP$. 
	
	Another simple, but important, property we require is a statement about the image of sets and their convex hulls 
	under linear functionals. Let $\sfS\subset\sfV$ be any subset of a vector space $\sfV$, then the {\em set-of-values over $\sfS$} 
	of a linear functional $\inprod{\calQ}{.}:\sfV\rightarrow\R$ satisfies 
	\[
		\Conv{ \lcb \inprod{\calQ}{X}; ~X \in \sfS \rom\rcb \rom} 
		~=~ 	
		 \lcb \inprod{\calQ}{X}; ~X \in \Conv{\sfS} \rom \rcb, 
	\]
	i.e. the convex hull of the set-of-values over $\sfS$  is the set-of-values over the convex hull of $\sfS$. 
	This statement follows immediately 
	from the linearity of the functional. In particular, the extrema of any  linear functional over non-convex sets are the 
	same as the extrema over their convex hulls 
	\be
		\inf_{X\in\sfS} \inprod{\calQ}{X} ~=~ \inf_{X\in\Conv{\sfS}} \inprod{\calQ}{X} , 
		\hstm \hstm 
		\sup_{X\in\sfS} \inprod{\calQ}{X} ~=~ \sup_{X\in\Conv{\sfS}} \inprod{\calQ}{X} . 		
	 \label{opt_conv_hull.eq}
	\ee
	This is useful when checking  extremal values of linear functionals over non-convex cones such as $\bbP_n^1$.

	We are now ready to state the linear-conic duality theorem, which is a minor modification of the general 
	topological vector space case as stated in~\cite{shapiro2001duality}.
	\begin{theorem}[Weak Linear-Conic Duality]									\label{conic_duality.thm} 
		Let $\cL:\sfV \rightarrow \sfV_\rme$  
		be a bounded  linear operator between  Banach spaces, and  let 
		$\sfP\subset\sfV$   be a (not necessarily convex) cone. 
		Denote by $\cL^\adj$, and $\sfP'\subset\sfV'$ the adjoint and dual cone respectively.  
		The following optimization 	problems are  duals  
        		\be
			\hspace*{-1em}
        			\inf_{\scriptsize \begin{tabular}{r} 	$\cL(X)=B$ \\ $X\in\sfP$  \end{tabular}} \inprod{Q}{X} 
        			\hstm \geq~ 
        			 \sup_{\scriptsize \begin{tabular}{r} $Q-\cL^\adj(Y)\in\sfP'$ \\ $Y\in V_\rme '$ \end{tabular}} \inprod{Y}{B}  		
        			 ~= 
        			 \ninf_{\scriptsize \begin{tabular}{r} $Q+\cL^\adj(Z)\in\sfP'$ \\ $Z\in V_\rme '$ \end{tabular}} \inprod{Z}{B} , 		
        		\ee 
                    If in addition there exists feasible $\Xba$ and $\Yba$ (or $\Xba$ and $\Zba$) that satisfy the 
                    		 complementary slackness (alignment) condition 
                    		\begin{align}
                    			\inprod{Q-\cL^\adj \big( \Yba \big) \rome }{\Xba} = 0
					\hstm \hstm 
					\big( \mbox{or}~\inprod{Q+\cL^\adj \big( \Zba \big) \rome }{\Xba} = 0	\big),
                    		  \label{align_cond.eq}
                    		\end{align} 
                    		then the two objectives are equal, and $\Xba$, $\Yba$ (or $\Xba$, $\Zba$) 
				are optimal for the respective problems. 
	\end{theorem} 
	\begin{proof} 		
		The following two functionals are equal whenever $X$ satisfies $B-\cL(X)=0$ 
		\begin{align*} 
			\inprod{Q}{X} 
				~&=~ \inprod{Q}{X} + \inprod{Y}{B-\cL(X) \romn} 
				  ~=~ \inprod{Q}{X} + \inprod{Y}{B} - \inprod{Y}{\cL(X) \romn}			\\
				  ~&=~ \inprod{Q}{X} + \inprod{Y}{B} - \inprod{\cL^\adj(Y)}{X \romn}	
				~=~  \inprod{Q-\cL^\adj(Y) \romn}{X} + \inprod{Y}{B} , 
		\end{align*} 
		where $Y\in\sfV_\rme'$ represents any linear functional. Thus their infima over 
		any subset of that set are equal. In particular 
		\be
			\inf_{\scriptsize	\begin{tabular}{r} $\cL(X)=B$ \\ $X\in\sfP$  \end{tabular}} 
				\inprod{Q}{X} 
			~=~ 
			\inf_{\scriptsize	\begin{tabular}{r} $\cL(X)=B$ \\ $X\in\sfP$  \end{tabular}} 
				\inprod{Q-\cL^\adj(Y) \romn}{X} + \inprod{Y}{B}.
		  \label{two_func_inf.eq}
		\ee
		Now if we restrict $Y$ such that $Q-\cL^\adj(Y)\in\sfP'$, then because $X\in\sfP$, we have
		$\inprod{Q-\cL^\adj(Y) \romn}{X} \geq 0$, and therefore 
		 a lower bound on the right hand side 
		\be
			\inf_{\scriptsize	\begin{tabular}{r} $\cL(X)=B$ \\ $X\in\sfP$  \end{tabular}} 
				\inprod{Q-\cL^\adj(Y) \romn}{X} + \inprod{Y}{B}
			~~\geq~~ 
				 \inprod{Y}{B}	, 
			\hstm 
			\mbox{for}~ Q-\cL^\adj(Y)\in\sfP'.
		  \label{two_func_ineq.eq}
		\ee
		Replacing the right hand side with its supremum over the set of constrained $Y$'s, and 
		replacing the left hand side with the left hand side of~\req{two_func_inf}, we get the desired inequality  
		\[
			\inf_{\footnotesize	\begin{tabular}{r} $\cL(X)=B$ \\ $X\in\sfP$  \end{tabular}} 
				\inprod{Q}{X} 			
			~~\geq~~ 
			\sup_{\footnotesize \begin{tabular}{r} $Q-\cL^\adj(Y)\in\sfP'$ \end{tabular}} \inprod{Y}{B} . 					
		\]
		Finally observe that if there is feasible pair $(\Xba,\Yba)$ satisfying the alignment condition~\req{align_cond}, 
		then the inequality~\req{two_func_ineq} becomes an equality, and the two problems have equal 
		 optimal values of the objectives. 
		 
		 The statement for the problem with the variable $Z$ follows from replacing $Y$ by $-Z$ in the 
		 supremum version of the dual 	 problem. 
	\end{proof}

%
%
%
%

\subsection{Schur Complements}

	Given any Hermitian matrix $M$, it defines a quadratic form  $\sfq(v):=\inprod{v}{Mv}$. The matrix is positive
	definite (semi-definite) if the quadratic form is positive (semi-definite) for all non-zero vectors. If the underlying 
	vector space is transformed by a linear transformation $v=Tw$, then 
	\[
		\sfq(v) ~=~ \inprod{v}{Mv} ~=~ \inprod{Tw}{MTw} ~=~ \inprod{w}{T^*MTw} .
	\]
	The transformation $M \mapsto T^*MT $ is a {\em congruence transformation} on the matrix $M$. 
	Thus under a change of basis, matrix representations of bilinear forms undergo congruence 
	transformations\footnote{This is in contrast to matrix representations of linear transformations, which undergo 
		similarity transformations.}. 
		Consequently, 
	congruence transformations preserve the sign definiteness of a Hermitian matrix. They also preserve its ``inertia'', 
	namely the triple of numbers indicating the number of negative, zero, and positive eigenvalues 
	respectively. 
	
	The well-known Schur complement can be understood as block-diagonalizatoin 
	by a congruence transformation on a $2\times 2$-block 
	partitioned matrix. Consider a Hermitian matrix $M$ partitioned as 
	\[
		M ~=~ \bbm M_{11} & M_\rmo \\ M_\rmo^* & M_{22} \ebm .  
	\]
	If $M_{22}$ is invertible, then $M$ can be block-diagonalized by the following congruence transformation 
	\be
		 \bbm M_{11} & M_\rmo \\ M_\rmo^* & M_{22} \ebm
		 ~=~ 
		 \bbm I & M_\rmo M_{22}^{-1} \\ 0 & I \ebm 
		 \bbm M_{11} - M_\rmo M_{22}^{-1} M_\rmo^* & 0 \\ 0 & M_{22} \ebm 
		 \bbm I & 0 \\ M_{22}^{-1} M_\rmo^* & I \ebm .
	  \label{Schur_BD.eq}
	\ee
	We can therefore immediately conclude that if $M_{22}>0$, then $M\geq 0$ iff the Schur complement 
	$M_{11}-M_\rmo M_{22}^{-1} M_\rmo^*\geq 0$. This is particularly useful if for example $M_{22}$ is fixed, 
	and the other blocks of $M$ contain a matrix variable such as might occur in a Linear Matrix Inequality (LMI). 
	In this case, positivity of the Schur complement gives another matrix inequality (a ``Riccati inequality'', 
	not an LMI) which can be more useful as we will 
	see in the next section.

\section{DLMIs, DRIs and DREs} 													\label{DLMIs.sec}

	In this section we explore the relationships between Differential Linear Matrix Inequalities (DLMIs), 
	Differential Riccati Inequalities (DRIs) and Differential Riccati Equations (DREs). 	
	To begin with, we recall some fundamental facts about the simplest of DLMIs, namely 
	Differential Lyapunov Inequalities, and then apply those to compare solutions of Riccati differential equations
	in terms of their coefficients' differences. 
	The exposition here is intended to be self contained.

	\subsubsection*{Differential Lyapunov Inequalities} 
	
	Consider a differential Lyapunov  inequality 
	of the form 
	\be
		\Xd + F^*X + X F ~\leq~ 0, 
		\hstm t\in[0,\sT], \hstm 
		\arraycolsep=2pt
		\begin{array}{rcl}  X(0) &=& 0, \\ \mbox{or}~X(\sT) &=& 0. \end{array} 
	  \label{Lyap_ineq_def.eq}
	\ee
	There are many solutions to such a differential inequality, but 
	they all have a sign definiteness depending on whether initial or final conditions are given. 
	The key to understanding the inequality~\req{Lyap_ineq_def}
	is to recast it  as an equality with a ``forcing term'' as follows 
	\begin{align}
		\Xd + F^*X + X F ~\leq~ 0
		\hstm &\Leftrightarrow \hstm 
		\Xd + F^*X + X F ~=~ -H, && H(t)\geq 0 	,				\nonumber		\\
		\hstm &\Leftrightarrow \hstm 
		-\Xd = F^*X + X F +  H  , && H(t)\geq 0 	,				\label{Diff_Lyap_Eq.eq}
	\end{align} 
	where $H$ is an arbitrary positive matrix-valued function. 
	Equation~\req{Diff_Lyap_Eq} is a differential Lyapunov equation 
	which has the following solutions\footnote{as can be verified by direct differentiation using Leibniz rule.} 
	depending on whether final or initial time conditions are given  
	\begin{align} 
		-\Xd = F^*X + X F +  H, \hsom X(0) = 0
			&&\Rightarrow \hstm 
			X(t) = - \smint{0}{t} \Phi^*(t,\tau) H(\tau) \Phi(t,\tau) ~d\tau , 				\label{Lyap_sol_H_n.eq}	\\
		-\Xd = F^*X + X F +  H, \hsom X(\sT) = 0
			&&\Rightarrow \hstm 
			X(t) =  \smint{t}{T} \Phi^*(t,\tau) H(\tau) \Phi(t,\tau) ~d\tau , 				\label{Lyap_sol_H_p.eq}
	\end{align} 
	where $\Phi$ is a state-transition matrix of $F$ such that  $(d/dt) \Phi(t,\tau) = \Phi(t,\tau) F(t)$. 
	
	Note that the definiteness of $X$ is determined by the definiteness of $H$ 
	in~\req{Lyap_sol_H_n} or~\req{Lyap_sol_H_p}. Combining this with~\req{Diff_Lyap_Eq}, we conclude that 
	\be
		\Xd + F^*X + X F ~\leq~ 0, 
		\hstm  
		\begin{array}{lcl} 
			X(0) = 0 & \hstm \Rightarrow \hstm  & X(t) \leq 0 , \hsom t\in[0,\sT],		\\
			X(\sT) = 0 & \hstm  \Rightarrow \hstm & X(t) \geq 0 , \hsom t\in[0,\sT]	.
		\end{array} 
	  \label{Lyap_diff_ineq_signs.eq}
	\ee 
	A useful intuition to keep in mind is that $\Xd\leq -F^*X-XF$ indicates that $X$ will start as negative if propagated forward 
	from $X(0)=0$. For backwards propagation from $X(\sT)=0$, the derivative ``looking backwards'' is $-\Xd$, and the 
	inequality $-\Xd \geq -F^*X-XF$ impies that $X$ should start as  positive when propagated backwards from zero final condition. 
	Finally, note that because a state-transition matrix is always full rank, the expressions~\req{Lyap_sol_H_n} 
	and~\req{Lyap_sol_H_p} imply inequalities can be replaced by strict inequalities in~\req{Lyap_diff_ineq_signs}.

	\subsubsection*{Differential Riccati Inequalities and Riccati Comparisons}

	By adding a quadratic and a constant term to the differential Lyapunov inequality, we obtain a 
	 Differential Riccati Inequality (DRI) which is of the form 
	\be
		\dot{\Lambda} + A^*\Lambda + \Lambda A 
			- \Lambda~ M~ \Lambda + Q   ~\geq~ 0, 
		\hstm \Lambda(\sT)=0.
	  \label{DRI_final.eq}
	\ee
	For  simplicity of exponsition, the
	 DRI here is stated for the inequality $\geq$ and a final condition. Similar DRIs can be obtained by using combinations 
	of $\geq,\leq$ and final or initial conditions. 
	
	The DRI~\req{DRI_final} has many solutions, but there is a maximal solution which has special properties as is described next. 
	In a manner similar to what was done for Lyapunov inequalities,  the DRI~\req{DRI_final} can be recast 
	 as a Differential Riccati Equation 
	(DRE) with an arbitrary positive  ``forcing'' term $H$ 
	\begin{align} 
		\dot{\Lambda} + A^*\Lambda + \Lambda A 
			- \Lambda~ M~ \Lambda + Q ~&=~ H    ~\geq~ 0, 
		\hstm \Lambda(\sT)=0, 												\nonumber		\\
		\Rightarrow \hstm 
		\dot{\Lambda} + A^*\Lambda + \Lambda A 
			- \Lambda~ M~ \Lambda + Q-H  ~&=~ 0, \hstm H\geq 0,  
		\hstm \Lambda(\sT)=0.												\label{DRE_H_pos.eq}
	\end{align} 
	Thus given any solution $\Lambda$  of the DRI~\req{DRI_final}, there exists a positive matrix-valued function $H$ such 
	that $\Lambda$ is a solution of the DRE~\req{DRE_H_pos} and vice versa.

	Some intuition for~\req{DRE_H_pos} can be obtained by rewriting it as follows 
	\[
		-\dot{\Lambda} ~=~  A^*\Lambda + \Lambda A 
			- \Lambda~ M~ \Lambda + Q-H,   \hstm H\geq 0,  
		\hstm \Lambda(\sT)=0.
	\]	
	Suppose the equation is solved ``backwards'' from the final condition over $[0,\sT]$. The quantity $-\Lambdad(t)$ 
	is the derivative ``looking backwards''. Compared with the choice $H(t)=0$, that backwards derivative can only 
	be smaller with a choice of $H(t)\geq0$.  
	Therefore the solution with $H(t)=0$ should be the largest of all solutions with $H(t)\geq0$. 
	Figure~\ref{DRE_DRI.fig} shows  simulations with 
	 a scalar example where the function $H$ was chosen as a positive random number
	switching at ten different points in an interval $[0,\sT]$. 
	One hundred such solutions are shown. In addition, the 
	solution with $H(t)=0$ is also shown, and it appears to be the maximum of all the other solutions. 
	We can show that this  behavior is true in general using the previous results on Lyapunov inequalities as shown next. 

	\begin{figure} 
		\centering 
		\includegraphics[width=0.35\textwidth]{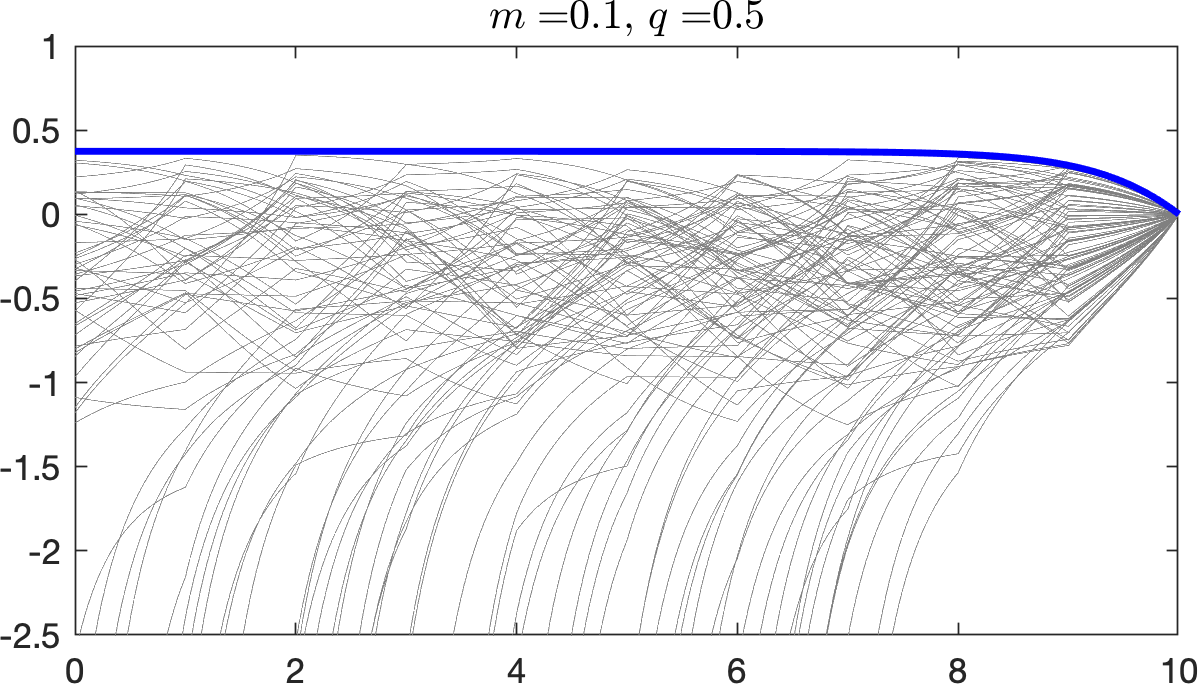}
		\quad\quad \quad\quad 
		\includegraphics[width=0.35\textwidth]{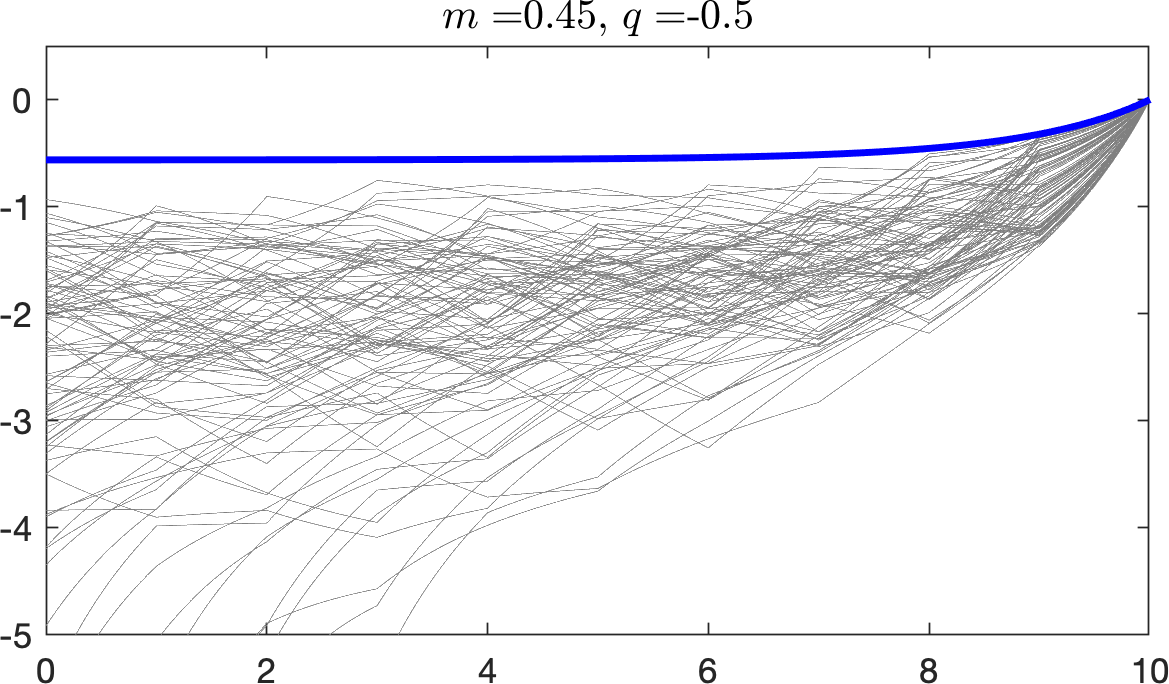} 		\\
		\includegraphics[width=0.35\textwidth]{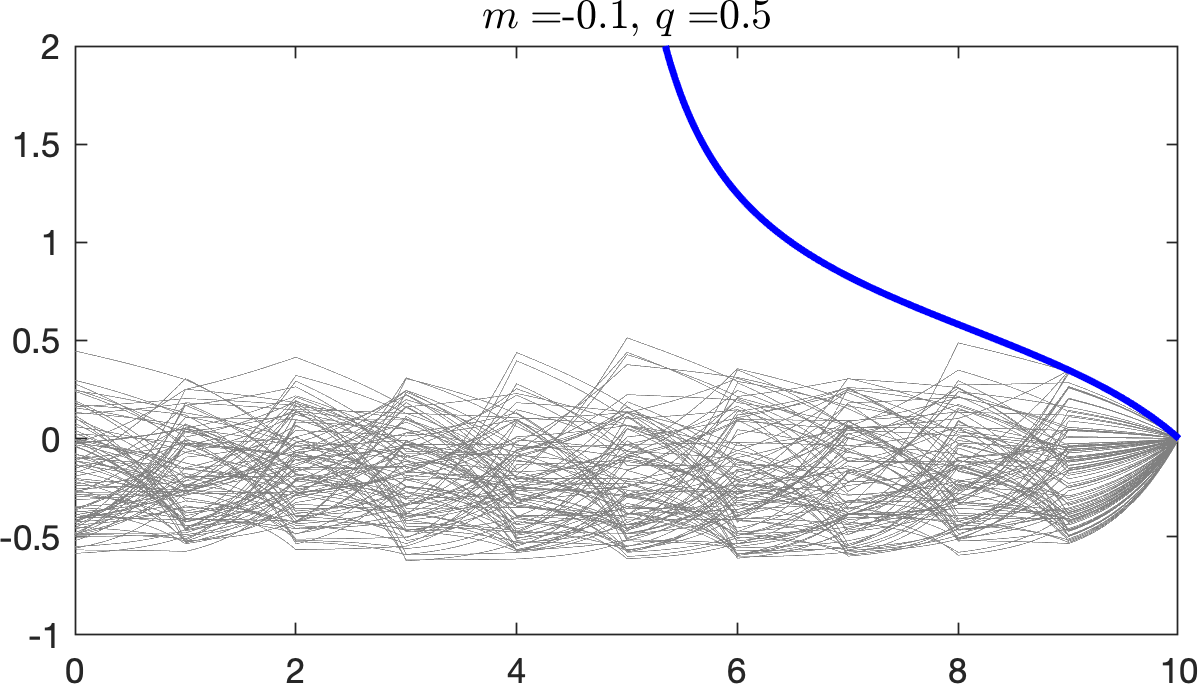}
		\quad \quad\quad\quad 
		\includegraphics[width=0.35\textwidth]{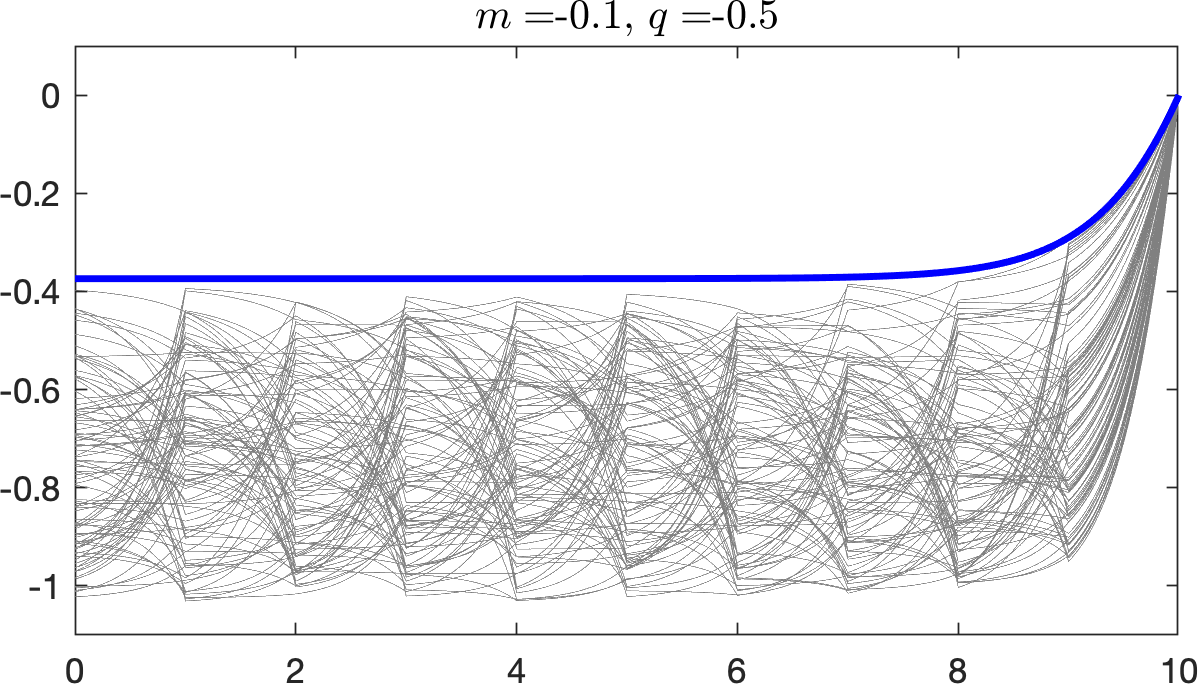}
		
		\mycaption{Examples demonstrating the maximal property of the solution $\Lambdab$ of the 
			Differential Riccati Equation (DRE) with final conditions~\req{DRE_max_lem}
			over all possible solutions of the 
			 Differential Riccati Inequality (DRI)~\req{DRE_lem}.
			 The equation here is scalar ($n=1$), and therefore becomes 
			 $-\dot{\lambdab}(t)= q+2a\lambdab(t) - m \lambdab^2(t)$, 
			 and various combinations of the signs of $q$ and $m$ are given. 
			 The gray curves represent 100 different solutions of the DRI, while the thick blue curve is the 
			 solution of the DRE. 
			 Notice how the maximality property is maintained even when the DRE and/or DRI solutions 
			 have finite escape time. 
			   } 
	  \label{DRE_DRI.fig}
	\end{figure} 

	\begin{lemma}[Maximal Solutions of DRIs] 								\label{DRI_DRE_only.lemma}
		Consider the final-condition Differential Riccati Inequality (DRI)
		\be
			\Lambdad + A^*\Lambda + \Lambda A -  \Lambda M \Lambda+ Q ~\geq ~ 0 , 
			\hstm\hstm \Lambda(\sT)=\Lambda_\rmf.										 \label{DRE_lem.eq}
		\ee																	
		The  maximal solution of the DRI is given by the 
		following Differential Riccati Equation (DRE) 
		\be
			\dot{\Lambdab} +
			A^*\Lambdab + \Lambdab A -  \Lambdab M	\Lambdab + Q ~=~0, 
				\hstm\hstm   \Lambdab(T)=\Lambda_\rmf   	,
		  \label{DRE_max_lem.eq}
		\ee 
		i.e. for any other $\Lambda$  satisfying the DRI,
		\[
			  \Lambdab(t) ~\geq~ \Lambda(t) , 			 			 
			 \hstm  t\in[0,\sT] . 
		\]
	\end{lemma}
	\noindent
		Note that this statement says nothing about  minimal 
		solutions to final value problems. 
		Indeed,  as shown by the examples in Figure~\ref{DRE_DRI.fig}, there may not exist 
		such solutions. 
	
		Before proving this statement, we recap a ``Riccati comparison'' result which is a 
                	 finite-horizon version of a classic  infinite-horizon argument~\cite[Lemma 3]{willems1971least}. 
                	Consider  solutions of two  Riccati differential equations with different ``$Q$-terms''
                	\be
                	\begin{aligned} 
                		0 ~&=~ \Lambdad_1 + A^*\Lambda_1 + \Lambda_1 A - \Lambda_1 M \Lambda_1 +Q_1 	,
                							\hstm \Lambda_1 (\sT) =\Lambda_{\rmf} ,									\\
                		0 ~&=~ \Lambdad_2 + A^*\Lambda_2 + \Lambda_2 A - \Lambda_2 M \Lambda_2 +Q_2	,
                							\hstm \Lambda_2 (\sT) =\Lambda_\rmf,
                	\end{aligned} 
                	\label{two_Ric_comp_lem.eq}
                	\ee
                	with $Q_1\geq Q_2$. 
                	Define the difference $\Lambdat := \Lambda_1 - \Lambda_2$, subtract the two equations and apply the 
                	Riccati difference formula~\req{Ric_diff} to obtain 
                	\begin{align*} 
                		\dot{\Lambdat}
                			+ \lb A+M\big( \Lambda_1+\Lambda_2 \big) /2 \rb^* \Lambdat  
                			+ \Lambdat \lb A+M\big( \Lambda_1+\Lambda_2 \big) /2 \rb
                			+ \lb Q_1 - Q_2 \rb    
                		~&=~ 0 	,
				\hstm\hstm \Lambdat(\sT)=0												\\
                		\Rightarrow \hstm 
                		\dot{\Lambdat}
                			+ \lb A+M\big( \Lambda_1+\Lambda_2 \big) /2 \rb^* \Lambdat  
                			+ \Lambdat \lb A+M\big( \Lambda_1+\Lambda_2 \big) /2 \rb
                		~&=~ - \lb Q_1 - Q_2 \rb    
                		~\leq~ 0 . 
                	\end{align*} 
                	This is a differential Lyapunov inequality of the form~\req{Lyap_ineq_def} for $\Lambdat$. 
                	Applying~\req{Lyap_diff_ineq_signs} we conclude that $\Lambdat(t)\geq0$, and  thus for the
                	differential Riccati equations~\req{two_Ric_comp_lem} with final conditions
                	\be
                		Q_1(t)  \geq Q_2(t), \hsom t\in[0,\sT]  
                		\hstm \Rightarrow \hstm 
                		\Lambda_1(t) \geq \Lambda_2(t) , \hsom t\in[0,\sT] .
                	  \label{Ric_comp_Q1Q2_lem.eq}
                	\ee
	Note that properties of $M$ (such as sign definiteness) play no role in the arguments above. 
 	\begin{proof}(of Lemma~\ref{DRI_DRE_only.lemma}).  	               	
               Consider the DRI~\req{DRE_lem} rewritten with the ``forcing variable'' $H\geq0$,
                and the DRE~\req{DRE_max_lem}  obtained from it by setting the inequality to equality
                	\begin{align*}  
                            \dot{\Lambda} + A^*\Lambda + \Lambda A 
                            			- \Lambda~ M~ \Lambda + Q   ~\geq~ 0 
			\hstm \Leftrightarrow \hstm 
                            \dot{\Lambda} + A^*\Lambda + \Lambda A 
                            			- \Lambda~ M~ \Lambda + (Q-H)    ~&=~ 0 , 	\hstm H\geq 0,
                            																\\
                            		\dot{\Lambdab} + A^*\Lambdab + \Lambdab A 
                            			- \Lambdab~ M~ \Lambdab + Q   ~&=~ 0. 
                	\end{align*} 
		Since $H\geq0$, then $Q\geq (Q-H)$, and the Riccati comparison result~\req{Ric_comp_Q1Q2_lem} implies that 
		\[
			\Lambdab(t) ~\geq~ \Lambda(t), 
			\hstm t\in[0,\sT]. 
			\qedhere
		\]
	\end{proof}

	\subsubsection*{Differential Linear Matrix Inequalities (DLMIs)}

	We are concerned with 
	 Differential Linear Matrix Inequalities (DLMIs) of a very special form  where a matrix-valued function $\Lambda(.)$
	is required to satisfy	
	\begin{align} 
		\calM(\Lambdablue) ~:=&~ \calQ + \cE^\adj(\dot{\Lambdablue}) + \cA^\adj(\Lambdablue) 		\nonumber	\\ 
		~:=&~ 
		\bbm Q & N  \\ N^*  & R  \ebm  
			+ 
	 	 \bbm I \\ 0 \ebm 
		 	{\color{blue}{\dot{\Lambda}}}
		 \bbm I & 0 \ebm 
			+
		 \bbm A^* \\ B^* \ebm \Lambdablue \bbm I & 0 \ebm 
		 	+ 
		 \bbm I \\ 0 \ebm \Lambdablue \bbm A & B \ebm 
																	\nonumber	\\
		= &~
		\bbm 	Q+\dot{\Lambdablue}+A^*\Lambdablue+\Lambdablue A & N+ \Lambdablue B \\ 
				 (N+ \Lambdablue B)^*  & R \ebm  
			~\geq~ 0 , 
			\hstm\hstm t\in[0,\sT]. 
 																	 \label{DLMI.eq}	 
	\end{align} 
	Even when boundary conditions for $\Lambda$ are specified, such DLMIs typically have non-unique 
	solutions. In other words, the inequality above does not provide sufficient constraints to uniquely 
	determined $\Lambdad(t)$ given $\Lambda(t)$, thus the non-uniqueness of solutions. 
	However, as in the case of Riccati inequalities, we are able to find maximal solutions of the DLMI by 
	relating it to a DRI via the Schur complement. 
%
%
%
%
%
	For notational simplicity,  the statements in this section are presented for the case $N=0$. The general case 
	is given in Theorem~\ref{DRI_DRE_app.lemma} of Appendix~\ref{DLMI_i_f.sec}. 
	
	Since the matrix $\calM(\Lambda)$ is in a $2\times2$-block partition, it is natural 
	  to use the Schur complement to characterize its positivity. 
	 First note that  $R\geq 0$ is a necessary condition for~\req{DLMI}. If we further assume 
	that $R>0$, then the Schur complement~\req{Schur_BD}  gives  the congruence 
	\be
		\bbm Q+\dot{\Lambda}+A^*\Lambda+\Lambda A & \Lambda B \\ 
				B^* \Lambda  & R \ebm  
		= 
		\bbm I & 	\Lambda B R^{-1} \\ 0 & I \ebm 	
		\bbm 	Q+\dot{\Lambda} + A^*\Lambda + \Lambda A  - \Lambda~ B R^{\sm1} B^*~ \Lambda	& 0 \\ 
				0																	& R \ebm 
		\bbm I & 0 \\  R^{-1} B^* 	\Lambda^* & I \ebm 	
																	\label{DRI_form.eq}
	\ee
	Thus the DLMI is equivalent to the DRI 
	\[
		Q+\dot{\Lambda} + A^*\Lambda + \Lambda A  - \Lambda~ B R^{\sm1} B^*~ \Lambda ~\geq~0. 
	\]
	Furthermore, if
	we have final conditions on the DLMI, then 
	the equivalent   DRI  is  of the form~\req{DRI_final}. 
	Therefore the previously stated facts about maximal solutions apply. 
	In addition, the left-hand-side of~\req{DRI_form} implies that the lowest rank possible for the DLMI matrix 
	is bounded from below by the rank of $R$. The right-hand-side of~\req{DRI_form} shows that this minimum 
	ranks is achieved 
	by setting the inequality in the DRI to equality (which renders the $(1,1)$ as zero). This gives a special 
	symmetric factorization of the DLMI matrix at the maximal solution as stated next.  
	\begin{lemma} (Maximal Solutions of the DLMI) 								\label{DRI_DRE.lemma}
		Consider the  final-value DLMI and its associated DRI 
		\[
            		\begin{aligned} 
            			\calM(\Lambda)
            				 &:=	 \calQ +\cE^\adj(\Lambdad) + \cA^\adj(\Lambda) 
            				 = 
            				\bbm 	Q+\dot{\Lambda}+A^*\Lambda+\Lambda A & \Lambda B \\ 
            				 B^* \Lambda  & R \ebm   
            				 \geq~ 0 , 										\\
            			\cR(\Lambda) &:= \Lambdad +
            			A^*\Lambda + \Lambda A -  \Lambda B  R^{\sm1} B^* \Lambda	+ Q ~\geq~ 0 ,
            		\end{aligned} 
			\hstm t\in[0,\sT], \hstm \Lambda(\sT)=\Lambda_\rmf, 
		\]
		where $R>0$. 
		A  maximal solution of either inequality  is given by the 
		 Differential Riccati Equation (DRE) 
		\be
			\dot{\Lambdab} +
			A^*\Lambdab + \Lambdab A -  \Lambdab B  R^{\sm1} B^* \Lambdab	+ Q ~=~0, 
				\hstm\hstm   \Lambdab(T)=\Lambda_\rmf   	,
		  \label{DRE_LQRf.eq}
		\ee 
		i.e. for any other $\Lambda$  satisfying the DLMI or equivalently the DRI, 
		$
			  \Lambdab(t) ~\geq~\Lambda(t) , 			 			 
			 ~ t\in[0,\sT] . 
		$
		Furthermore, at the maximal solution, the matrix $\calM(\Lambdab)$ has the 
		full-rank symmetric factorization 
		\be
			\calM(\Lambdab) 
			~=~ 
			\bbm \Lambdab B R^{\sm\frac{1}{2}}  \\  R^{\frac{1}{2}} \ebm  
			\bbm  R^{\sm\frac{1}{2}} B^* \Lambdab &  R^{\frac{1}{2}} \ebm . 
		  \label{DLMI_rank_min.eq}
		\ee
	\end{lemma}
	\begin{proof} 
		Note that the equivalences of the DREs and the DLMIs follow from the 
		Schur complements  for  $R>0$. Thus showing the maximality property for the DRIs (already shown in 
		Lemma~\ref{DRI_DRE_only.lemma})
		proves the maximality property for the DLMIs. 

		For the factorization~\req{DLMI_rank_min}, note that if $\Lambdab$ which satisfies the DRE is 
		used in both sides of~\req{DRI_form} 
            	\begin{align*}
                        	  \hspace*{-1em} 
            		\bbm Q+\dot{\Lambdab} + A^*\Lambdab + \Lambdab A & \Lambdab B \\  B^* \Lambdab &  R \ebm 
            		~&= ~
			\bbm \Lambdab B R^{-1} \\ I \ebm R \bbm R^{-1} B^* \Lambdab^* & I \ebm 
			~=~
            		\bbm \Lambdab B R^{\sm\frac{1}{2}}  \\  R^{\frac{1}{2}} \ebm  
            			\bbm  R^{\sm\frac{1}{2}} B^* \Lambdab B &  R^{\frac{1}{2}} \ebm , 
		\end{align*} 
		which is a full-rank factorization since $R$ (and therefore $R^{\frac{1}{2}}$) is assumed non-singular. 		
	\end{proof} 
	
	\noindent
	A few remarks are needed to emphasize the implications of the above result. 
	\begin{itemize} 
		\item 
                    	The symmetric factorization~\req{DLMI_rank_min}
                    	implies that the rank of $\calM(.)$ 
                    	at the extremum $\Lambdab$ is exactly equal the rank of $R$. 
                    	Since we have assumed $R$ is non-singular, 
                    	then this is  the lowest rank $\calM(.)$ can be.  
                    	We therefore conclude that the DRE solution 
                    	also {\em minimizes the rank} of the DLMI matrix $\calM(\Lambda)$. 
                    	This fact  will be important for understanding the 
                    	structure of optimal  signals later on, as it will imply the specific state feedback form of optimal inputs.  
		\item 
                    	Another important point to reiterate here is that one needed to convert a DLMI to a DRI in order to find extremal 
                    	solutions by substituting equality for inequality in the DRIs. 
                    	It is not possible to find maximal solutions of the DLMIs  directly by changing their inequalities to equalities. 
                    	 The DLMI with equality is infeasible (it would imply for example that $R=0$). The 
                    	DLMI and the DRI are equivalent by the Schur complement, but changing the inequalities to equalities 
                    	does not maintain 
                    	that equivalence. Thus to find maximal solutions to a DLMI, one must convert it to the equivalent DRI, 
                    	and then set the latter
                    	to equality. 
	
	\end{itemize}

	\subsubsection*{Symmetric Factorizations and the Lur'e Equations} 
	
	As an aside, we point out that the
	 factorization~\req{DLMI_rank_min} indicates another route that can be taken to arrive at the DRE starting from the 
	DLMI {\em without using the Schur complement}. 
	Observe that a matrix  $M$ is positive iff it has a symmetric factorization of the form $M=UU^*$ for some matrix $U$. 
	Consider any such factorization of $\calM(\Lambda)$ which is partitioned conformably as 
	\be
		\calM(\Lambda) ~:=~ 
		\bbm Q+\dot{\Lambda} + A^*\Lambda + \Lambda A & (N+\Lambda B) \\ (N+ \Lambda B)^* &  R \ebm 
		~=~ 
		\bbm U_1 \\ U_2 \ebm \bbm U_1^* & U_2^* \ebm  ~=:~ UU^*. 
	  \label{DLMI_fact.eq}
	\ee
	This gives the following constraints on the submatrices
	$U_1$ and $U_2$ 
	\be
		U_1U_1^* = Q+\dot{\Lambda} + A^*\Lambda + \Lambda A , 
		\hstm 
		U_1U_2^* = N+\Lambda B , 
		\hstm 
		U_2U_2^* = R  . 
	  \label{DLMI_fact_pieces.eq}
	\ee
	These equations are the time-varying version of what are 
	 sometimes referred to as Lur'e equations, which are normally stated as follows. Given 
	$A,B,Q,N,R$, find matrices $\Lambda$, $U_1$ and $U_2$ such that the equations~\req{DLMI_fact_pieces} hold.  
	In our current language, this amounts to finding $\Lambda$ such that $\calM(\Lambda)\geq0$,  as well as  a
	symmetric factorization of $\calM(\Lambda)$. Such symmetric factorizations always exist for any positive 
	matrix $\calM(\Lambda)$.

	If $R>0$, then the factorization $R=U_2U_2^*$ implies that $U_2$ has full row rank. The ``tightest'' such 
	factorization would have $U_2$ to be square. This corresponds to the factorization~\req{DLMI_fact} being 
	full rank, which is also the minimal rank that the factor $U$ can have. 
	Thus, if $R>0$, and  we adopt a full-rank factorization in~\req{DLMI_fact}, then 
	 $U_2$ is square and invertible, which gives $U_1=\Lambda B U_2^{-*}$, and in turn 
	\[
		U_1U_1^* ~=~ \Lambda B U_2^{-*} U_2^{-1} B^* \Lambda 
		~=~ \Lambda B \lb U_2 U_2^*\rb^{-1}  B^* \Lambda 
		~=~ \Lambda B R^{\sm1}  B^* \Lambda .
	\]
	Substituting this in the first equation in~\req{DLMI_fact_pieces} gives the DRE 
	\[
		\Lambda B R^{-1}  B^* \Lambda = Q+\dot{\Lambda} + A^*\Lambda + \Lambda A.
	\]
	Thus we see that {\em the DRE arises out of a full-rank factorization of the DLMI matrix}~\req{DLMI_fact}, 
	or equivalently, a minimal-rank (of $U$) solution to the Lur'e equations. 
	The Schur complement is not used in this  argument. 

\section{The Deterministic  LQR Problem} 									\label{LQR_det.sec}

	Now we apply the procedure outlined in Section~\ref{outline.sec} to the classic Linear Quadratic Regulator (LQR) 
	problem in its deterministic form. Note that in the finite time-horizon case, all 
	 statements given below are applicable to time-invariant or 
	time-varying systems and performance objectives provided that the matrix-valued function $\calQ(.)$ is bounded.  
	For notational simplicity again, we consider the case with no cross terms $N=0$. The general case
	follows from Theorem~\ref{IQC_i.thm} in a later section. 
	\begin{theorem} 														\label{LQR_i.thm}
		Consider  a linear (possibly time-varying) system of the form 
		\be
			\xd  ~=~ A x ~+~ Bu , 
			\hstm\hstm 
			t\in[0,\sT], \hstm 		x(0) =  \rmx_\rmi , 
		 \label{LQR_system.eq}
		\ee
		and a quadratic form defined on $(x,u)$ pairs 
		\[
			\qform(x,u) 
			~:=~\int_{0}^{\ssT} \bbm x \\ u \ebm^* \! \bbm Q & 0 \\ 0 & R \ebm  \bbm x \\ u \ebm dt 
			~=:\int_{0}^{\ssT} \bbm x \\ u \ebm^* \! \bigmat{\calQ}   \bbm x \\ u \ebm dt 
			~=:\int_{0}^{\ssT} {\sf q}(x,u) ~dt, 
		\]
		with $Q\geq0$ and $R>0$. 
		
		The infimum  of the quadratic form $\qform$ subject to the constraints~\req{LQR_system} is 
                        		\be
                        			\inf_{x,u} ~\qform (x,u) ~=  ~\rmx_\rmi^* ~\Lambdab(0) ~\rmx_\rmi , 
                        		\label{IQC_infima_i.eq}
                        		\ee
                        		where $\Lambdab$ is the maximal solution of the Differential Linear Matrix Inequality (DLMI) over $[0,\sT]$
            		\be
            			\calM(\Lambda) :=
            			\calQ + \bbm \Lambdad + A^*\Lambda  + \Lambda A &  \Lambda B \\  B^*\Lambda & 0 \ebm
            			~\geq~ 0 , 
            			\hstm 
            						\Lambda(\sT) = 0 .
            		  \label{LQR_thm_DLMI.eq}
            		\ee
		The  maximal solution of this DLMI  is given by the solution of 
                        		 the Differential Riccati Equation (DRE) 
                                    	\[
                                    		Q+\dot{\Lambdab} + A^*\Lambdab + \Lambdab A 
                                    			-\Lambdab B R^{\sm1}  B^*\Lambdab = 0, 
                                    		\hstm \hstm 
                                    		\Lambdab(\sT) = 0.
                                     	\]
		The optimal control $\ub$ is given by the state feedback law 
                        		\[
                        			\ub(t) = \lb R^{\sm1}(t) B^*(t) \Lambdab(t) \rom\rb ~\xb(t) , 
					\hstm t\in[0,\sT]. 		
                        		\]
	\end{theorem}


	As outlined in the introduction,  we will give the proof of this Theorem in several steps. 
	Step (1)  is to reformulate 
	the minimization problem in terms of covariance matrices, which yields a linear-linear optimization problem with a 
	cone constraint. 
	Step (2)  is to use  conic duality to reformulate it as a maximization problem 
	with a 
	Differential Linear Matrix Inequality (DLMI) constraint. 
	Step (3) is to use the maximal
	 solution of the DLMI (which is given by a  DRE) to solve the dual problem.  Finally, step (4) uses the alignment (complementary 
	 slackness) conditions to give the optimal primal solution in the form of a state feedback and show that the duality gap is zero. 
	
\begin{enumerate} 
	\item 	
	The first step is to 
	define the  joint (deterministic)  covariance of $x$ and $u$ 
	\be
		\Sigma(t) ~:=~ \bbm x(t) \\ u(t) \ebm \bbm x^*(t) & u^*(t) \ebm , 
	 \label{Cov_def_det.eq}
	\ee
	which is always of rank one. 
	The differential equation for $\Sigma$  is~\req{Sigma_diffeq_LQP}, which we recap for clarity and rewrite it in 
	the abstract integral form~\req{int_eq_EAX}  developed in Section~\ref{adjoint_systems.sec} 
	\begin{align} 
	 	\cE \big( \Sigma \big) := \Em {\Sigma} \Ems , 
		\hstm &
		\cA \big(\Sigma  \big) := \bbm A & B \ebm {\Sigma} \Ems
			+ \Em  {\Sigma}   \bbm A^* \\ B^* \ebm  		,							\nonumber		\\
			\arraycolsep=2pt
			\left. 
			\begin{array}{rcl} 
	 			\cE\big(\Sigmad\big) 	&=& \cA\big( \Sigma \big) , 	\\
				\cE\big( \Sigma(0) \big) 	&=& \rmx_\rmi \rmx_\rmi^* =: \sfX_\rmi  
			\end{array}	\right\} 
			\hstm &\Leftrightarrow \hstm 
			\lbb \cE  - \intf\cA  \rbb (\Sigma)  ~=~   \sfX_\rmi. 								 \label{Sig_diff_int.eq}
	 \end{align} 
	 Now the LQR problem can be written in the covariance representation  as 
	 \be
	 	\boxed{	
		\begin{aligned} 
		\inf_{\Sigma\in\LOne(0,T) }  	& \int_0^T \trcc{ \calQ ~\Sigma } ~dt 					
					\hspace{10em} 	& \fbox{\sf Primal:  LQR}		\\		
	 		\big(  \cE - \intf\cA \big) (\Sigma) &= \sfX_\rmi			\\
				\Sigma \geq 0, ~\rank{\Sigma} &=1. 
		\end{aligned} ~
		}
	  \label{Sigma_opt_two_det.eq}
	 \ee

	 \item 
	 
	 This problem is in the form to which the  linear-conic duality Theorem~\ref{conic_duality.thm} can be applied. 
	The primal and dual problems are then  
	\be
		\arraycolsep=1em
		{\sf Primal:} 
		\left\{ 
		\begin{array}{ccc} 
			{\arraycolsep=2pt
			\begin{array}{rcl} 
				\displaystyle 
				\inf_{\Sigma\in\LOne(0,\ssT)}  & &  \inprod{\calQ}{\Sigma} 		\\ 
				\big(  \cE - \intf\cA \big) (\Sigma) &=& \sfX_\rmi		\rule{0em}{1.5em} 	\\
				\Sigma \geq 0, ~\rank{\Sigma} &=&1 
			\end{array} 
			}
			&   \geq    &  
			{\arraycolsep=2pt
			\begin{array}{rcl} 
				\displaystyle 
				\sup_{Y\in\LInf(0,\ssT)}	 & &  \inprod{Y}{\sfX_\rmi} 		\\ 
				\calQ - \big(  \cE^\adj - \cA^\adj\intb \big) (Y) &\geq& 0		\rule{0em}{1.5em} 
			\end{array} 
			}
		\end{array}  
		\right\}{\sf : Dual} , 
	 \label{LQR_prim_dual.eq}
	\ee
	where we recall the calculation of the adjoint of the linear operator in the equality constraint
	\be
		\lbb \cE - \intf\cA  \rbb^\adj
		~=~ 
		\cE^\adj - \cA^\adj   \intf^\adj
		~=~ 
		\cE^\adj - \cA^\adj \intb . 
	 \label{lin_equal_adj.eq}
	\ee
	
	Note that the rank condition does not have any effect on the dual problem. 
	This is of course because the cone dual to positive rank-one matrices $\bbP_n^1$ is the convex cone 
	$\bbP_n$ of positive matrices, i.e. 
	the dual problem is always 
	convex even if the primal problem is not. Furthermore, since the cost objective is linear, the argument in~\req{opt_conv_hull} 
	implies that the extrema are achieved at the ``vertices'' of the convex cone.  

	Examining the dual objective $\inprod{Y}{\sfX_\rmi}$, we see that since $\sfX_\rmi$ is a constant function over 
	$[0,\sT]$, the objective  can be simplified by introducing the backwards  integral of
	$Y$ as follows 
	\[
		\inprod{Y}{\sfX_\rmi} 
		~=~ \smint{0}{T} \trcc{ Y(t)~\sfX_\rmi \rom} ~dt 
		~=~\trcc{ \smint{0}{T}  Y(t) dt~ \sfX_\rmi} 
		~=:~ \trcc{ \Lambda(0)  ~\rmx_\rmi \rmx_\rmi^*  \rom} 
		~=~ \rmx_\rmi^* \Lambda(0) ~\! \rmx_\rmi  , 
	\]
	where we  defined $\Lambda$ as the backwards integral of $Y$ 
	\be
		\Lambda:=\intb Y
		\hstm \Leftrightarrow \hstm 
		\arraycolsep=2pt
		\left\{
		\begin{array}{rcl} 
			\Lambda (t) &	:=&  \int_t^T Y(\tau)~d\tau 		\\ 
			-\dot{\Lambda} (t) &:=&  Y(t)
		\end{array} 
		\right\}
		\hstm \Leftrightarrow \hstm 		
		\Lambda(0)=\int_0^T Y(t)dt	, \hsom \Lambda(\sT) = 0 . 
	  \label{Lambda_def.eq}
	\ee
	
	The dual constraint in terms of $\Lambda=\intb Y$ becomes 
	\[
		\calQ - \big(\cE^\adj - \cA^\adj \intb \big) (Y) 
		~=~ 
		\calQ - \cE^\adj(Y) + \cA^\adj(\intb Y ) 
		~=~  
		\calQ + \cE^\adj(\Lambdad) + \cA^\adj(\Lambda ) , 
		\hstm 
		\Lambda(\sT) = 0. 
	\]
	This is actually a  DLMI for  $\Lambda$ 
           \begin{align} 
      		\calQ + \cE^\adj \lbb \Lambdad \rbb + \cA^\adj\blb \Lambda \rbb 
            		 ~&=~
			\bbm Q & 0 \\ 0 & R \ebm 
			+ \bbm I \\ 0 \ebm \dot{\Lambda} \bbm I & 0 \ebm 
			+ \bbm A^* \\ B^* \ebm \Lambda \bbm I & 0 \ebm 
			+ \bbm I \\ 0 \ebm \Lambda \bbm A & B \ebm 
																	\nonumber	\\
            		 &=~ 
			\bbm Q+\dot{\Lambda} + A^*\Lambda + \Lambda A & \Lambda B \\ B^* \Lambda &  R \ebm
			~=:~ \calM(\Lambda) .
																\label{dual_eq_DLMI.eq} 
           \end{align} 
	The abstract dual problem~\req{LQR_prim_dual} can now be written concretely in terms of $\Lambda$ as 
	 \be
	 	\boxed{	~
		\begin{aligned} 
		\sup_{\Lambda\in \LInf(0,T)} ~  \rmx_\rmi^* \Lambda(0)~\!\rmx_\rmi						&	
				\hspace{7em}	& \fbox{\sf Dual:  LQR}		\\	
		\calM(\Lambda):=	
		\bbm Q+\dot{\Lambda} + A^*\Lambda + \Lambda A & \Lambda B \\ B^* \Lambda &  R \ebm
		~&\geq~ 0 , \hstm  
				   &
				 t\in[0,\sT] , ~~\Lambda(\sT) = 0,						
		\end{aligned} ~
		}
	  \label{LQR_det_dual.eq}
  	\ee
	Note that while the objective $\rmx_\rmi^* \Lambda(0)~\!\rmx_\rmi$ is quadratic in $\rmx_\rmi$, it is {\em linear} 
	in the optimization variable $\Lambda$.

	\item 
	The dual problem is a maximization problem with a final condition on a DLMI constraint. We can therefore 
	use the maximal solution $\Lambdab$  of the DLMI (Lemma~\ref{DRI_DRE.lemma}) to obtain the function 
	$\Lambdab(.)$ that achieves the supremum. 
	$\Lambdab$ satisfies the DRE 
	\[
		-\dot{\Lambdab} ~=~ 
		A^*\Lambdab + \Lambdab A -\Lambda BR^{-1} B^* \Lambda + Q , 
		\hstm 
		\Lambdab(\sT)=0. 
	\]
	The maximality 
	property implies that for any other  $\Lambda$ satisfying  the DLMI constraint 
	\[
		\Lambdab(t) ~\geq~ \Lambda(t)
		\hstm \Rightarrow \hstm 
		\Lambdab(0) ~\geq~ \Lambda(0)
		\hstm \Rightarrow \hstm 
		\rmx_\rmi^* \Lambdab(0) ~\! \rmx_\rmi ~\geq ~ \rmx_\rmi^* \Lambda(0) ~\! \rmx_\rmi.
	\]
	Therefor $\Lambdab$ solves the dual problem~\req{LQR_det_dual}.

	\item 
	It remains to check the alignment condition. The optimal $\Sigmab$, if it exists must satisfy the alignment 
	condition~\req{align_cond}, which in this case is 
	\[
		0~=~ 
		\inprod{\calQ - \big(\cE^\adj - \cA^\adj \intb \big) (\bar{Y})}{\rom \Sigmab } 
		~=~ 
		\inprod{\calM(\Lambdab)}{\rom  \Sigmab }
	\]
	 Furthermore, $\Sigmab$ must be rank one, which means 
	\be
		0 =  \inprod{\calM(\Lambdab) \rom}{  \Sigmab}
		 =  \inprod{\bbm \Lambdab BR^{\sm\frac{1}{2}}  \\ R^{\frac{1}{2}} \ebm  
					\bbm R^{\sm\frac{1}{2}} B^* \Lambdab & R^{\frac{1}{2}}  \ebm }
					{  \bbm \xb \\ \ub \ebm \bbm \xb^* & \ub^* \ebm},  
	   \label{align_det_LQR.eq}
	\ee
	where the last expression comes from the full-rank factorization property~\req{DLMI_rank_min} of the DRE. 
	Now observe that~\req{align_det_LQR} is a statement that two positive matrices are orthogonal. Symmetric 
	factorizations of the two orthogonal matrices are given, and therefore Lemma~\ref{pos_orth_function.lemma} implies 
	that if such a $\Sigmab$ existed, then 
	\begin{align} 
		\bbm R^{\sm\frac{1}{2}} B^* \Lambdab & R^{\frac{1}{2}}  \ebm   \bbm \xb \\ \ub \ebm  ~=~ 0
		\hstm &\Leftrightarrow \hstm 
		R^{\sm\frac{1}{2}} B^* \Lambdab ~\xb ~+~  R^{\frac{1}{2}} \ub ~=~ 0 				\label{align_one.eq}	\\
		&\Leftrightarrow \hstm 
		\ub(t) ~=~ \lb -R^{-1} B^* \Lambdab\rb ~\xb(t).  								\label{align_two.eq}
	\end{align} 
	The argument then goes as follows. If we choose $\ub$ according to~\req{align_two}, then~\req{align_one} holds, 
	and in turn the joint covariance $\Sigmab$ of $(\xb,\ub)$ is of rank one and 
	 satisfies the alignment condition~\req{align_det_LQR}. This proves that the duality gap is zero.
	 The alignment condition therefore forces the  optimal control signal $\ub$ to be
	  in the form of a state feedback on the optimal state trajectory $\xb$. 
	This is the well-known classical solution to the  LQR problem. 

	Finally, the optimal cost is given from $\Lambda(0)$ as 
	\[
		\qform_{\rm opt}(\rmx_\rmi)  
			~=~  \rmx^*_\rmi  \Lambda(0) \rmx_\rmi  .
	\]
	This is the so-called ``value function'' of this optimal control problem, 
	which is a quadratic form  in the initial conditions of the LQR problem. 
	Note that while the optimal feedback gain $-R^{\sm1}B^* \Lambdab$ is independent of the initial condition, 
	the optimal trajectory $\xb$, control $\ub$ and cost $\rmx_\rmi^*\Lambda(0)\rmx_\rmi$ are functions of the initial condition. 
	
\end{enumerate}

\section{The Stochastic LQR Problem} 									\label{LQR_stoch.sec}

	The stochastic version of the LQR is not to be confused with the Linear Quadratic Gaussian (LQG) problem, 
	in which measurements are partial and noisy. The Stochastic LQR  (SLQR) problem assumes 
	 additive stochastic disturbances in the state equation that play the role of perturbing  the equilibrium state, 
	just like the   non-zero initial condition in the deterministic 
	LQR problem perturbs the equilibrium. 
	As we will see, the solution to this problem is exactly the same as the deterministic one, which is the LQR 
	state feedback.  
	
	 The problem can be stated as
	follows. Consider a (possibly time-varying)  system with a random initial condition and  
	driven by zero-mean white noise uncorrelated with past histories of $x$ and $u$ 
	\be
		\xd = Ax + Bu + w , 
		\hstm 
			\arraycolsep=2pt
			\begin{array}{rcl} 
				\expct{w(t)w^*(\tau) \romn} &=& \delta(t-\tau) ~\sfW , \\ 
				\expct{w(t)x^*(\tau) \romn} &=& 0 , ~t\geq \tau , \\ 
				\expct{w(t)u^*(\tau) \romn} &=& 0 , ~t\geq \tau ,
			\end{array} 
		\hstm 
			\begin{array}{rcl} 
										x(0) &=& \rmx_\rmi, 		\\			
				\expct{\rmx_\rmi \rmx_\rmi^* \romn} &=& \sfX_\rmi, 		\\ 
				\expct{\rmx_\rmi w^*(t)} &=& 0, 
			\end{array} 
	  \label{LQR_stoch_dyn.eq} 		
	\ee
	where $\sfW$ and $\sfX_\rmi$ are the covariance matrices of $w$ and $x(0)$ respectively. 
	The goal is to find the control input $u$ that minimizes the 
	following quadratic  performance objective
	\be
		\qform(x,u) ~:=~  \smint{0}{T} \expct{\rom x^*Qx +2x^*Nu+ u^*Ru } dt . 
	  \label{LQR_stoch_J.eq}
	\ee
	The expectation is taken over the joint
	 probability distribution of $w$ and $x(0)$, and therefore $\qform$ is an ``average'' over all 
	realizations of the  process $w$ and random variable $x(0)$. 
	That is why $w$ and $x(0)$ are not included as  arguments of $\qform$. After taking expectation, 
	$\qform$ is a function of only $x$ and $u$. 
	
	As before, $\qform$ can be rewritten as a linear functional on the joint covariance of $(x,w)$, which is now an actual covariance
	 matrix of stochastic processes
	\begin{align} 
		\qform ~&:= 
			 \int_0^T  \expct{ \bbm x^* & u^* \ebm \bbm Q & N \\ N^* & R \ebm \bbm x \\ u \ebm } dt	
		      =
			 \int_0^T  \expct{ \trcc{ \bbm Q & N \\ N^* & R \ebm \bbm x \\ u \ebm  \bbm x^* & u^* \ebm }} dt	 \nonumber\\
		      &=
			 \int_0^T \trcc{   \bbm Q & N \\ N^* & R \ebm~ \expct{ \bbm x \\ u \ebm  \bbm x^* & u^* \ebm }} dt 
			 ~=: \int_0^T \trcc{\calQ \, \Sigma \rom } dt 
			 ~= \inprod{\calQ}{\Sigma\rom} , 
	  \label{LQR_stoch_obj1.eq} 
	\end{align} 
	where 
	\[
		\Sigma
		~:=~ \expct{ \bbm x  \\ u \ebm  \bbm x^* & u^* \ebm } 
		~=~ \bbm \expct{xx^* }   &    \expct{xu^*} \\   \expct{ux^*} &   \expct{uu^*} \ebm  
		~=:~ \bbm \Sigma_{xx} & \Sigma_{xu} \\ \Sigma_{ux} & \Sigma_{uu} \ebm.
	\]
	Keep in mind that $\Sigma$ is now a deterministic function which is to be chosen as a the solution of a (deterministic) 
	optimization problem. 

	The differential equation for  $\Sigma_{xx}$ is obtained from  a standard 
	 calculation 
	\be
		\Sigmad_{xx} 
			~=~ 
		 			A~\Sigma_{xx}~+ ~B~\Sigma_{ux} ~+~ \Sigma_{xx} ~A^*~+~\Sigma_{xu} ~B^*~+~ \sfW. 
	  \label{SLQR_Sigmaxx_diffeq.eq} 
	\ee
	Using our earlier notation for the matrix operators $\cE$ and $\cA$, this differential equation  can be written 
	as follows together with its integral form 
	\be
            	\cE\big( \Sigmad \big) 
            			= \cA\big( \Sigma  \big) + \sfW , \hstm  \cE\big( \Sigma(0) \big) = \sfX_\rmi,
            	\hstm \Leftrightarrow \hstm 
            	\lbb \cE  ~-~\intf \cA  \rbb (\Sigma)  ~=~ \intf \sfW ~+~ \sfX_\rmi .
	  \label{Sigma_diffeq.eq}
	\ee
%
	We can now state the primal  stochastic LQR problem abstractly as 
	 \be
	 	\boxed{	~
		\begin{aligned} 
		\inf_{\Sigma }  	&~ \int_0^T \trcc{ \calQ ~\Sigma } ~dt 					
					\hspace{10em} 	& \fbox{\sf Primal:  Stochastic LQR}		\\		
	 		\big(  \cE -\intf \cA \big) (\Sigma) &= \intf \sfW +  \sfX_\rmi			\\
				\Sigma &\geq 0 
		\end{aligned} ~
		}
	  \label{Sigma_opt_two.eq}
	 \ee

	 Let's compare this problem statement with the deterministic LQR version~\req{Sigma_opt_two_det}. If 
	 $\sfW=0$, then the two statements are identical except for the rank-1 constraint on $\Sigma$. Since this 
	 constraint does not effect the dual problem, the duals of deterministic and stochastic LQR (with $\sfW=0$) 
	 should be identical. The only difference would be in the alignment condition, which we will explore shortly. 

	The linear operator in the equality constraint is exactly the same as in the deterministic LQR problem, and its
	adjoint has already been calculated in~\req{lin_equal_adj}.  Linear-conic duality thus gives the following 
	statement 
	\be
		\arraycolsep=1em
		{\sf Primal:} 
		\left\{ 
		\begin{array}{ccc} 
			{\arraycolsep=2pt
			\begin{array}{rcl} 
				\displaystyle 
				\inf_{\Sigma}  & &  \inprod{\calQ}{\Sigma} 		\\ 
				\big(  \cE - \cA\intf \big) (\Sigma) &=& \intf \sfW+ \sfX_\rmi		\rule{0em}{1.5em} 	\\
				\Sigma  &\geq&0 
			\end{array} 
			}
			&   \geq    &  
			{\arraycolsep=2pt
			\begin{array}{rcl} 
				\displaystyle 
				\sup_{Y}	 & &  \inprod{Y}{\intf\sfW+\sfX_\rmi}		\\ 
				\calQ - \big(  \cE^\adj - \cA^\adj\intb \big) (Y) &\geq& 0		\rule{0em}{1.5em} 
			\end{array} 
			}
		\end{array}  
		\right\}{\sf : Dual} , 
	 \label{SLQR_prim_dual.eq}
	\ee	
	The only difference between the dual here and the dual in the deterministic problem is the cost objective, 
	which again can be rewritten in terms of $\Lambda$, the backwards integral of $Y$ defined in~\req{Lambda_def} 
	 \begin{align*} 
	 	\inprod{Y}{\intf\sfW+\sfX_\rmi \romn} 
		~=~ 
	 	\inprod{\intb Y}{\sfW} +\inprod{Y}{\sfX_\rmi} 
		~=~ 
	 	\inprod{\Lambda}{\sfW} + \trcc{ \smint{0}{T}   Y(t) ~\sfX_\rmi dt }			
		~&=~ 
	 	\inprod{\Lambda}{\sfW} + \trcc{ \smint{0}{T}  Y(t) dt ~\sfX_\rmi }				\\
		&=~  
		\inprod{\Lambda}{\sfW}+\inprod{\Lambda(0)}{\sfX_\rmi} .
	 \end{align*} 
	 Note that the first functional $\inprod{\Lambda}{\sfW}$ is between functions on $[0,\sT]$, while the 
	 second $\inprod{\Lambda(0)}{\sfX_\rmi}$ is between matrices.

	Recalling the expression~\req{dual_eq_DLMI} for the dual inequality constraint in terms of $\Lambda$, 
	the dual problem can now be stated as 
	 \be
	 	\boxed{	~
		\begin{aligned} 
		\sup_\Lambda ~~ \trcc{\rom \Lambda(0) \sfX_\rmi} +  \smint{0}{T} \trcc{\rom \Lambda \sfW}dt 				
								&&		 \fbox{\sf Dual: Stochastic LQR}		\\		
		\bbm Q+\dot{\Lambda} + A^*\Lambda + \Lambda A & \Lambda B \\ B^* \Lambda &  R \ebm
		~&\geq~ 0 , \hstm\hstm 
				  & 
				 t\in[0,\sT] , ~~\Lambda(T) = 0.					
		\end{aligned} ~
		}
 										 \label{Lambda_opt_ineq.eq}	 
  	\ee
	Although the objective of this problem has two terms, each of those terms is maximized by the maximal 
	solution to the DLMI. Since the DLMI has a final condition, we know this maximal solution is the solution 
	of the DRE 
	\[
		A^*\Lambdab + \Lambdab A -\Lambdab BR^{-1} B^* \Lambdab + Q ~=~0 , 
		\hstm\hstm  
		\Lambdab(T)=0. 
	\]

	The alignment for this problem requires a little more care and interpretation than the deterministic LQR problem. 
	Let $\Sigmab$ be a candidate joint covariance for the optimal $\Sigma$. In the stochastic setting, the rank one 
	constraint is not imposed, and $\Sigmab$ may have higher rank. To investigate this, let 
	\be
		\Sigmab(t) ~:=~ \expct{ \bbm \xb(t)  \\ \ub(t) \ebm  \bbm \xb(t)^* & \ub(t)^* \ebm } ~=~ U(t) ~U^*\!(t) 
		 	~=~ \bbm U_1(t) \\ U_2(t) \ebm \bbm U_1^*(t) & U_2^*(t) \ebm , 			
	\ee 
	be a symmetric decomposition, where $U$ is partitioned conformably with  $(x,u)$.  This decomposition does not have
	to be full rank for the arguments that follow. 
%
%
%
%
%
	As described in Appendix~\ref{Cov.subsec}, this
	 decomposition allows for writing the  two random processes $\xb$ and $\ub$ as a linear transformation on 
	a   random process $v$ with uncorrelated components (i.e. $\expct{v(t)v^*(t)}=I$)  
	\be
		\bbm \xb(t) \\ \ub(t) \ebm ~=~ \bbm U_1(t) \\ U_2(t) \ebm v(t) 
	  \label{xu_from_v.eq}
	\ee
	Now recall the alignment condition~\req{align_cond}, which in this case reads   
	\be
		0 = \inprod{\Sigmab}{\calQ+\cE^*(\dot{\Lambdab})+\cA^*(\Lambdab)} 
			= \inprod{  \bbm U_1 \\ U_2 \ebm \bbm U_1^* & U_2^* \ebm}
					{\bbm \Lambdab BR^{\sm\frac{1}{2}}  \\ R^{\frac{1}{2}} \ebm  
					\bbm R^{\sm\frac{1}{2}} B^* \Lambdab & R^{\frac{1}{2}}  \ebm }, 
	   \label{align_stoch_LQR.eq}
	\ee
	where we again used the full-rank factorization property~\req{DLMI_rank_min} of the DRE solution. 
	Lemma~\ref{pos_orth_rank.lemma} on mutually orthogonal positive matrices states that 
	\begin{align*} 
		0 ~=~ 
		\bbm R^{-\frac{1}{2}} B^* \Lambdab & R^{\frac{1}{2}}  \ebm 
		\bbm U_1 \\ U_2 \ebm
		\hstm &\Leftrightarrow \hstm 
		R^{-\frac{1}{2}} B^* \Lambdab~ U_1 ~+~  R^{\frac{1}{2}}~U_2 ~=~ 0 			\\
		&\Leftrightarrow \hstm 
		U_2 ~=~ -R^{-1} B^* \Lambdab~ ~U_1 .
	\end{align*} 
	This last equation, together with~\req{xu_from_v} gives a relation between the optimal $\xb$ and $\ub$  
	\be
		\ub(t) ~=~ U_2(t) ~v(t) ~=~ -R^{-1} B^* \Lambdab(t) ~ ~U_1(t)  ~v(t) 
		~=~  \lb-R^{-1} B^* \Lambdab(t) \rb~\xb(t) . 
     	  \label{stoch_LQR_optK.eq} 
	\ee
	Thus the optimal control $\ub$ is the same  static state feedback as the deterministic LQR problem. 
	The difference in this case is that the optimal covariance $\Sigmab$ is not necessarily of rank one, but possibly 
	higher. Equation~\req{align_stoch_LQR} (together with the rank bound of Lemma~\ref{pos_orth_rank.lemma})
	 implies that the rank of $\Sigmab$ can be at most $n$ (the state dimension), 
	and may be lower depending on the rank of $\sfW$ (this last statement is not a consequence of~\req{align_stoch_LQR}).
	
	Note again that the optimal feedback  control law~\req{stoch_LQR_optK} does not depend on either the initial 
	condition $\sfX_i$ or the disturbance $\sfW$ covariance matrices. The optimal cost value however depends on those 
	\[
		\qform_{\rm opt} 
		~=~ \trcc{\rom \Lambdab(0) \sfX_\rmi} +  \smint{0}{T} \trcc{\rom \Lambdab(t)  ~\sfW(t)}dt
	\]

\section{General Integral Quadratic Constraints (IQCs)} 								\label{IQC.sec} 

	We now  consider the case of general Integral Quadratic Constraint (IQCs). 
	 In general, these problems concern a state space model with an input, and a
	 functional jointly quadratic in input and state. The problem is to characterize the
	 set of values of this quadratic form 
	evaluated over the set of all possible trajectories generated by all possible inputs. 
	Due to the linearity of the system's equations, this set of values can only be of three types $[\gamma,\infty)$, 
	$(-\infty,-\gamma]$ or $(-\infty,\infty)$ for some number $\gamma\geq 0$. 
	The problem is then reduced to computing the infimum of the quadratic form subject to the system's equations, 
	then this infimum will be either $-\infty$ or  a finite number $\pm\gamma$. 	
	A geometric illustration of such LQP problems in general is given in Figure~\ref{LQ_w_constraints.fig}. 
	For example, in
	 the $\Hinfty$-norm  and passivity checking problems, the initial condition is typically assumed to be zero, 
	 and the  problems are thus  reduced to checking whether the value set 
	 of a quadratic form is $[0,\infty)$. This property holds for the system if the respective infimum is zero, 
	 and does not hold if it is $-\infty$.

\begin{figure}[t]
	\centering
	\includegraphics[width=0.4\textwidth]{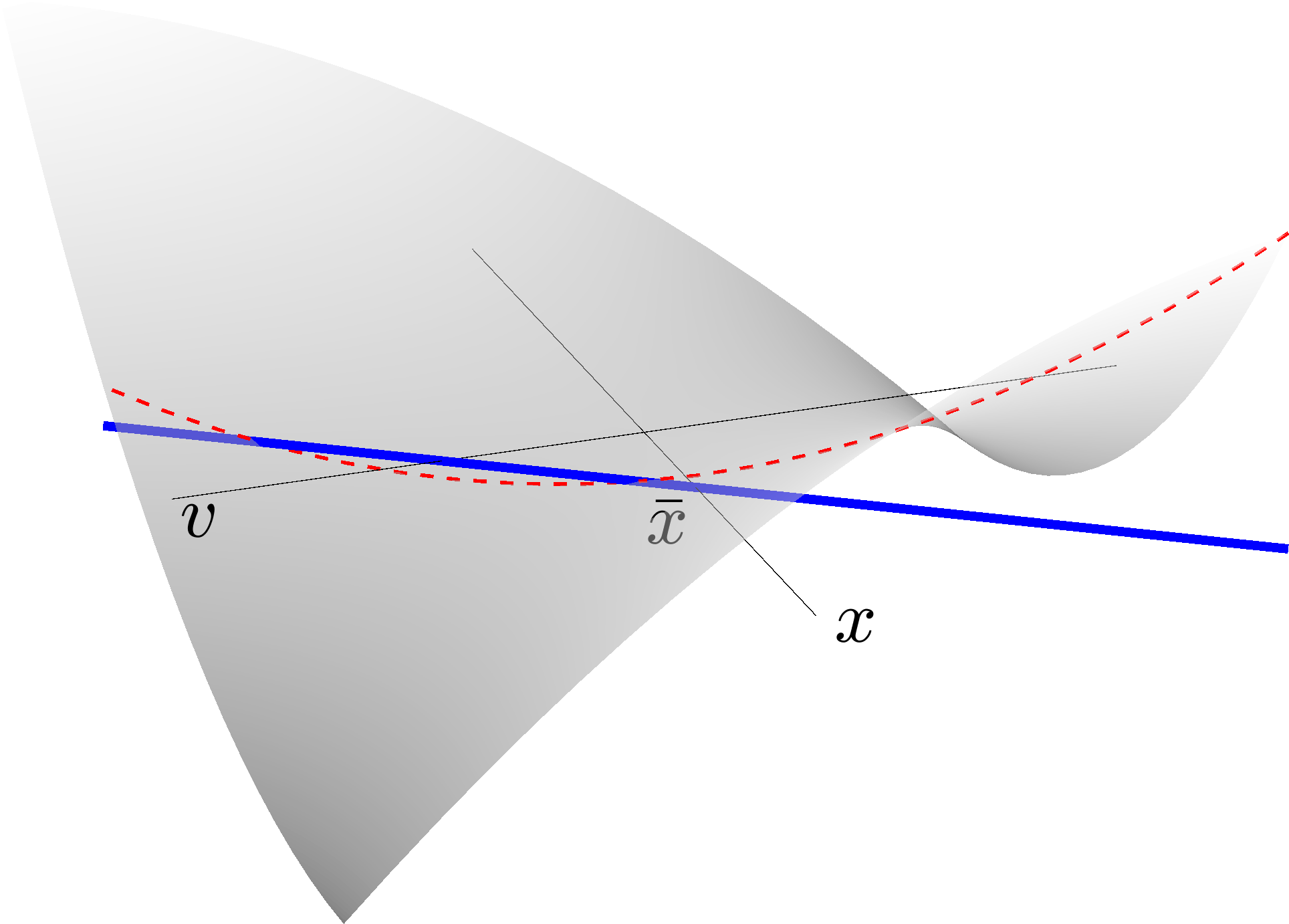} 
	
	\mycaption{Depiction of a quadratic form on input-state pairs $(v,x)$. The blue line represents the linear-affine space 
		of all input-state pairs satisfying the dynamics $\dot{x}=Ax+Bv$, $x(0)=\rm\rmx_\rmi$. When the input $v=0$, the
		state trajectory is the initial condition response $\bar{x}(t) = e^{At} \rm\rmx_\rmi$, depicted here as the intersection of 
		the linear-affine space with the ``$x$-axis''. The quadratic form $\qform(x,v)$ over all signal pairs (not necessarily 
		constrained by the dynamics) is depicted as the grey surface. $\qform$ can have mixed signature as depicted 
		here. The LQ Problem is to determine the infimum of the values of the quadratic form over the linear-affine constraint
		set (those values depicted here as the dashed red curve). Whether the infimum is finite or $-\infty$ depends on both 
		the quadratic form and the system dynamics. In the case where the initial condition is zero ($\xb=0$), 
		the constraint set is a 
		subspace, and the constrained infimum can only be either $0$ or $-\infty$. 
			} 
  \label{LQ_w_constraints.fig} 
\end{figure}


	\begin{theorem} 														\label{IQC_i.thm}
		Consider  a linear (possibly time-varying) system of the form 
		\be
			\xd  ~=~ A x + B v , 
			\hstm\hstm 
			t\in[0,\sT], \hstm 		x(0) =  \rmx_\rmi , 
		 \label{IQC_system_i.eq}
		\ee
		and a quadratic form defined on $(x,v)$ pairs 
		\[
			\qform(x,w) 
			~:=~\int_{0}^{\ssT} \bbm x \\ v \ebm^* \! \bbm Q & N \\ N^* & R \ebm  \bbm x \\ v \ebm dt 
			~=:\int_{0}^{\ssT} \bbm x \\ v \ebm^* \! \bigmat{\calQ}   \bbm x \\ v \ebm dt 
			~=:\int_{0}^{\ssT} {\sf q}(x,v) ~dt, 
		\]
		with $Q=Q^*$ and $R>0$. 
		Consider also the Differential Linear Matrix Inequality (DLMI) over $[0,\sT]$
		\be
			\calM(\Lambda) :=
			\calQ + \bbm \Lambdad + A^*\Lambda  + \Lambda A &  \Lambda B \\  B^*\Lambda & 0 \ebm
			~\geq~ 0 , 
			\hstm 
						\Lambda(\sT) = 0 .
		  \label{IQC_thm_DLMI.eq}
		\ee
		\begin{enumerate} 
			\item 
                        		The infimum  of the quadratic form $\qform$ subject to the constraints~\req{IQC_system_i} is 
                        		\be
                        			\inf ~\qform (x,w) ~=  ~\rmx_\rmi^* ~\Lambdab(0) ~\rmx_\rmi , 
                        		\label{IQC_infima_i.eq}
                        		\ee
                        		where $\Lambdab$ is the maximal solution of the DLMI~\req{IQC_thm_DLMI}. This maximal solution 
                        		satisfies the Differential Riccati Equation (DRE) 
                                    	\[
                                    		Q+\dot{\Lambdab} + A^*\Lambdab + \Lambdab A 
                                    			- \lb N+\Lambdab B\rb R^{-1} \lb N+\Lambdab B\rb^* = 0, 
                                    		\hstm \hstm 
                                    		\Lambdab(\sT) = 0.
                                     	\]
			\item 
                        		If the solution of the DRE escapes to infinity within $[0,\sT]$, then the  infimum 
                        		in~\req{IQC_infima_i} is $-\infty$. Otherwise, the  optimal signal $\vb$ is given by the state feedback
                        		\[
                        			\vb = R^{-1} (N^*+ B^* \Lambdab) ~\xb, 		
                        		\]
		\end{enumerate} 
	\end{theorem}

	\begin{proof} 	
	Since this is an infimization problem with  fixed initial conditions, all the steps of the proof are essentially the 
	same as the LQR case. The differences being that there are no assumptions on the definitness of $Q$, 
	the additional  term $N$ in $\calQ$, and the interpretation of the joint 
	covariance\footnote{As stated in the introduction, we require $R>0$. In fact, we can assume either 
		$R<0$ or $R>0$, but we assume here $R>0$ as a convention. The case where $R$ has mixed signature (as in e.g. 
		$\sfH^\infty$ state feedback design) requires a different treatment.}.
	 We therefore only point out the differences in this proof. 
	
       	The deterministic joint covariance of $x$ and $v$  is defined as 
            	\[
            		\Sigma(t) ~:=~ \bbm x(t) \\ v(t) \ebm \bbm x^*(t) & v^*(t) \ebm . 
            	\]
		The primal and dual problems read exactly like the LQR case~\req{LQR_prim_dual}. The DLMI constraint of the dual 
		in this case becomes 
             	\be
			0~\leq~ 
              		\calQ 
            		+ 
            		\bbm \Lambdad +A^*\Lambda +\Lambda A & \Lambda B \\ B^* \Lambda & 0 \ebm  
			= 
            		\bbm Q+ \Lambdad +A^*\Lambda +\Lambda A & N+\Lambda B \\  (N+\Lambda B)^* & R \ebm  , 
            		\hstm \Lambda(\sT) = 0	.
            	  \label{Mpm_DLMI_i.eq}
		\ee
            	Finally, the  dual problem is therefore  
            	\be
            			\displaystyle
            			\sup_{	\calM(\Lambda)	\geq 0, ~\Lambda(T)=0	} 
            				\rmx_\rmi^* 	\Lambda(0) \rmx_\rmi 	
            	  \label{IQC_dual_DLMI_i.eq}
            	\ee

            	The solution of this  last problem is obtained  from the slightly more general   
            	 Theorem~\ref{DRI_DRE_app.lemma} of Appendix~\ref{DLMI_i_f.sec}, 
	 	which states that the maximal solution of the DLMI~\req{Mpm_DLMI_i}  
            	is given by the solution of the DRE
            	\[ 
            		Q+\dot{\Lambdab} + A^*\Lambdab + \Lambdab A 
            			- \lb N+\Lambdab B\rb R^{-1} \lb N+\Lambdab B\rb^* ~=~ 0,	
            		\hstm 
            		\Lambdab(\sT) = 0, 
            	\] 
		where the  maximality property implies 
            	\[
            		\Lambdab(t) ~\geq~ \Lambda(t) 
            		\hstm \Rightarrow \hstm 
            		\rmx_\rmi^*\Lambdab(0)~\! \rmx_\rmi ~\geq~ \rmx_\rmi ^* \Lambda(0) ~\! \rmx_\rmi  
            		\hstm \Rightarrow  
            		\displaystyle
            			\sup_{	\calM(\Lambda)	\geq 0, ~\Lambda(T)=0	} 
            				\rmx_\rmi^* 	\Lambda(0) ~\! \rmx_\rmi 	
            	 	=~ \rmx_\rmi^*  \Lambdab(0) ~\!\rmx_\rmi	.
            	\]
		In this problem, the  alignment condition~\req{align_cond} of conic duality reads 
            	\begin{align*} 
            		0 ~&=~ 
            		\inprod{\calQ-(\cE^*-\cA^*\intb)(\bar{Y}) \rom}{\Sigmab} 
            		~=~ 
            		\inprod{\calQ+\cE^*(\dot{\Lambdab})+\cA^*(\Lambdab) \rom}{\Sigmab} 			\\
            		&=~ 
            		\inprod{\bbm (N+\Lambdab B) R^{\sm\frac{1}{2}}  \\  R^{\frac{1}{2}} \ebm  
            			\bbm  R^{\sm\frac{1}{2}} (N+ \Lambdab B)^* &  R^{\frac{1}{2}} \ebm }
				{\bbm \xb \\ \vb \ebm \bbm \xb^* & \vb^* \ebm},
            	\end{align*} 
            	where the last equality follows from the full-rank symmetric factorization~\req{DLMI_rank_min_app} 
            	of $\calM(\Lambdab):=\calQ+\cE^*(\dot{\Lambdab})+\cA^*(\Lambdab)$ at the extremal $\Lambdab$. 
            	Lemma~\ref{pos_orth_rank.lemma} implies that  the product of the respective symmetric factors must be zero 
            	\[
            		0 ~=~ 
            		\bbm  R^{\sm\frac{1}{2}} (N+ \Lambdab B)^* &  R^{\frac{1}{2}}	\rom \ebm 	
            			 \bbm \xb \\ \vb \ebm
            		\hstm 	\Leftrightarrow 	\hstm 
            		\vb ~=~ -R^{\sm1} (N^* + B^* \Lambdab) ~\xb.
            	\] 
            	Thus the optimal $\vb$ is in the form of a state feedback. 
	\end{proof}

	\subsection{$\sfL^2(0,\sT)$-induced Norm, the Bounded-Real Lemma}
	
	Consider the (possibly time-varying) system with zero initial conditions
	\be
	\begin{aligned}
		\xd ~&=~ A x ~+~ B w , 	\hstm t\in[0,\sT],  \hstm x(0)=\rmx_\rmi = 0 , 		\\ 
		z ~&=~ C x 
	\end{aligned} 
	\label{ss_IQC_examples.eq}
	\ee
	This system has 
	  $\sfL^{2}[0,\sT]$-induced norm less than $\gamma>0$ iff
	  the IQC with quadratic form~\req{Hinf_IQC_set} 
	\[
		\qform(x,w) ~:=~ \smint{0}{T} \sfq(x,w) dt
			~:=~		
		 \int_{0}^{T} 	
		\arraycolsep=1.5pt
		 \bbm x \\ w \ebm^* \bbm \sm C^*C & 0 \\ 0 & \gamma^2 I  \ebm \bbm x \\ w \ebm dt
	\]	
	is positive. Applying Theorem~\ref{IQC_i.thm} we see that 
	\[
		\inf \qform(x,w) ~=~ \rmx_\rmi \Lambda(0) \rmx_\rmi ~=~ 0 
		\tag{since $\rmx_\rmi=0$}
	\]
	iff the DRE 
	\[
		\Lambdad + A^*\Lambda + \Lambda A - \frac{1}{\gamma^2} \Lambda B B^* \Lambda - C^*C ~=~ 0, 
		\hstm 
		\Lambda(T) = 0, 
	\]
	has a (bounded) solution over the interval $[0,\sT]$. Otherwise the infimum is $-\infty$. 
	The corresponding DLMI is given by  
	\[
			\calQ + \bbm \Lambdad + A^*\Lambda  + \Lambda A &  \Lambda B \\  B^*\Lambda & 0 \ebm
			~=~ 
			\bbm \Lambdad + A^*\Lambda  + \Lambda A - CC^*  &  \Lambda B \\  B^*\Lambda & \gamma^2 I \ebm
			~\geq~ 0 , 
			\hstm 
			\Lambda(T) = 0. 
	\]
	
	Note that since in this problem the constant term in the DRE is $Q-NR^{\sm1} N^* = -C^*C\leq 0$, then by 
	the monotonicity properties (Lemma~\ref{DRE_monotone.lemma}) of the DRE, we have $\Lambda(t)\leq 0$. 
	Furthermore, if $(C,A)$ is observable, then $\Lambda(t)<0$ for $t\in[0,\sT)$. 
	
	For future reference, if we assume the 
	 infinite-horizon version of this problem is arrived at by setting $\Lambdad=0$ in the above, then we have the 
	 statement 
	\[
		\|G\|_{2-\rmi} ~\leq~ \gamma 
		\hstm \Rightarrow \hstm 
		\exists \Lambda \leq 0, 
		~
		\bbm  A^*\Lambda  + \Lambda A - CC^*  &  \Lambda B \\  B^*\Lambda & \gamma^2 I \ebm
		~\geq~ 0 .
	\]
	If we replace $\Lambda$ by $-X$, then the statement is equivalent to 
	\[
		\|G\|_{2-\rmi} ~<~ \gamma 
		\hstm \Rightarrow \hstm 
		\exists X\geq 0, 
		~
		\bbm  A^*X  + X A + CC^*  &  X B \\  B^*X & -\gamma^2 I \ebm
		~\leq~ 0 ,
	\]
	which is the more conventional statement of the Bounded Real Lemma. The reverse statement is easy to establish by 
	a completion-of-squares argument. In addition, 
	if $(C,A)$ is observable, 
	 then  we can replace the condition $X\geq0$
	above with $X>0$. These arguments will be presented elsewhere for the infinite-horizon case.

	\subsection{Passivity over $(0,\sT)$, the Positive-Real Lemma}
	
	Appendix~\ref{HinfPass.appen} explains that the system 
	\be
	\begin{aligned}
		\xd ~&=~ A x ~+~ B w , 	\hstm t\in[0,\sT],  \hstm x(0)=\rmx_\rmi = 0 , 		\\ 
		z ~&=~ C x ~+~ D w ,
	\end{aligned} 
	\label{ss_IQC_examples.eq}
	\ee
	 is  passive over $[0,\sT]$ iff the quadratic 
	form~\req{Passivity_IQC_set}  
	\[
			\qform(x,w) ~:=~ \int_0^T  \bbm x \\ w \ebm^* \bbm 0 & C^* \\ C & D+D^* \ebm \bbm x \\ w \ebm dt 
	\]
	is positive. Theorem~\ref{IQC_i.thm} states that
	\[
		\inf \qform(x,w) ~=~ \rmx_\rmi^* \Lambdab(0) \rmx_\rmi ~=~ 0 
	\]
	(i.e. the quadratic form constrained to the dynamics is positive) 
	iff the the following DRE 
	\begin{align*} 
		0~&=~ \Lambdad + A^*\Lambda + \Lambda A - (C^*+ \Lambda B) (D+D^*)^{\sm1}  (C+B^* \Lambda),	
		\hstm 
		\Lambda(\sT) = 0, 
	\end{align*} 
	has a (bounded) solution over the interval $[0,\sT]$. Otherwise the infimum is $-\infty$. 
	The associated DLMI is 
	\[
		\bbm \Lambdad + A^*\Lambda + \Lambda A  & C^*+\Lambda B \\ C + B^* \Lambda & D+D^* \ebm ,
		\hstm 
		\Lambda(\sT) = 0.
	\]
	For this quadratic form, $N=C^*$, $Q= 0$, and therefore $Q-NR^{\sm1}N^* = -C^*(D+D^*)^{\sm1} C\leq 0$. 
	The DRE definiteness 
	Lemma~\ref{DRE_monotone.lemma} states that the solution of the DRE is decreasing going backwards. 
	Thus $\Lambda(t)\leq0$ for $t\in[0,\sT]$. 
	
		For future reference, if we assume the 
	 infinite-horizon version of this problem is arrived at by setting $\Lambdad=0$ in the above, then 
	the infinite-horizon version of the problem is
	\[
		G ~\mbox{passive} 
		\hstm \Rightarrow \hstm 
		\exists \Lam \leq 0, ~ 
		\bbm A^*\Lam + \Lam A  & C^*+P B \\ C + B^* P & D+D^* \ebm
		~\geq~ 0. 
	\]	
	In addition, if we  again assume $(C,A)$  observable, then 
	we can replace $\Lam\leq 0$ above by $\Lam<0$.  If we define $P:=-\Lambda$, and multiply the LMI 
	by a minus sign we get the statement 
	\[
		G ~\mbox{passive} 
		\hstm \Leftrightarrow \hstm 
		\exists P \geq 0, ~ 
		\bbm A^*P + P A  & P B-C^* \\  B^* P-C & -D-D^* \ebm
		~\leq~ 0, 
	\]	
	which is the traditional statement of the ``Positive Real Lemma''. Note the the converse implication follows from a
	standard ``completing the squares'' argument.

\appendix

\section{Appendices}

\subsection{Trace Duality} 											\label{trace_duality.sec}

Denote by $\sigma_i(A)$ the $i$'th singular value of a matrix $A$ when arranged in descending order. First note that 
for a square matrix $A$, we have that $\left| \trcc{A} \right| \leq \sum_{i=1}^n \sigma_i(A)$. Indeed, let $A=U\Sigma V^*$ be the singular 
value decomposition\footnote{Here
			 $\Sigma=\diag{\sigma_1,\ldots,\sigma_n}$ is the diagonal matrix of singular values of $A$, 
			  not to be confused with the covariance matrix 
			$\Sigma$ defined in the remainder of this paper.} of $A$, then 
\[
	\left| \trcc{A} \right| ~=~ \left| \trcc{U\Sigma V^*} \right|  ~=~ \left| \trcc{\Sigma V^*U}  \right| 
	~=~  \sum_{i=1}^n \sigma_i(A) ~\left| \alpha_{i} \right|  ~\leq~ \sum_{i=1}^n \sigma_i(A),
\]	
where $\lcb \alpha_i \rcb$ are the diagonal elements of $V^*U$ which are all such that $|\alpha_i|\leq 1$ since $V^*U$ is unitary. 

Now let $M$ and $H$ be any two real matrices of the same dimensions 
\[
	\trcc{H^*M} 
	~\leq~ \sum_{i=1}^n \sigma_i\big( H^*M \big) 
	~\leq~ \sum_{i=1}^n \sigma_{\max}(H)~  \sigma_i\big(M \big) 
	~=~ \|H\|_\infty ~ \sum_{i=1}^n  \sigma_i\big(M \big) 
	~=~ \|H\|_\infty \|M\|_1 .
\]
This inequality is tight in the sense that for any $H$
\[	
	\sup_{\|M\|_1\leq 1} ~\trcc{H^*M} ~=~ \|H\|_\infty.
\]
To prove this, let $H=U\Sigma V^*$ be a singular value decomposition 
\[
		\trcc{H^*M	\romn} 
		~=~ \trcc{ V\Sigma U^* ~M	\romn} 
		~=~ \trcc{ \Sigma ~U^* \!M V	\romn} , 
\]  
Now choose $M$ such that  $U^*MV= \diag{1,0,\ldots,0}$. Since the $\|.\|_1$ norm is unitarily invariant, then $\|M\|_1 = \|U^*MV\|_1=1$, 
and 
\[
	 \trcc{ \Sigma ~ \diag{1,0,\ldots,0}	\romn} ~=~ \sigma_{\max}(H) ~=~ \|H\|_\infty. 
\]

\subsection{$\Hinfty$ Norms and  Passivity}  						\label{HinfPass.appen}

	Given a linear  time-invariant system $G$,  its $\Hinfty$ norm is the induced norm when the $\sfL^2$ norm 
	on signals is used. It can also be  characterized in the frequency domain by maximizing the singular values
	of the frequency response over all frequencies 
	\[
		\|G \|_{2-\rmi} ~:=~ \sup_{0\neq w\in \sfL^2} \frac{\| Gw\|_{\sfL^2} }{\|w\|_{\sfL^2} }
			~=~ \|\Gh\|_\infty ~:=~ \sup_{\omega\in \R} ~\sigmax \lb\rom \Gh(j\omega ) \rb  , 
	\]
	where $\Gh$ is the transfer function representing the system $G$. 
	The frequency domain characterization only applies to time-invariant (and therefore infinite time horizon) systems. 
	It is useful to have a time-domain criterion for the $\sfL^2$-induced norm of time-varying and/or finite-horizon 
	systems. 
	This means characterizing the  $\sfL^2[0,\sT]$-induced norm, which we now recast 
	a a linear-quadratic optimization problem of the form~\req{LQP}. 

           Consider  the state-space (possibly time-varying) system 
           \be
           \begin{aligned} 
           	\xd ~&=~ Ax ~+~ B w ,	
				\hstm\hstm t\in[0,\sT], \hstm x(0)=0, 				\\ 
		z ~&=~ Cx ~+~ D w , 
           \end{aligned} 
           \label{system_LQP_IQC.eq}
           \ee
           and the    following implications regarding its (finite or infinite-horizon)  $\sfL^2$-induced norm 
                \[  
                		\hspace{-1em} 
                		\begin{array}{rclrcl}
                        \|G\|_{2-\rmi} ~\leq~ \gamma  & \Longleftrightarrow &
                        		\displaystyle
                                    	\sup_{0\neq w\in \sfL^2}  \frac{\| z \|^2_2}{\| w \|^2_2}  ~\leq~\gamma^2 		& &		\\
                                                 &   \Updownarrow     &      									& &			 \\
                                   	 \forall ~{w\in \sfL^2}, ~w\neq 0,
                                    & &       \displaystyle
						\frac{\| z \|^2_2}{\| w \|^2_2}
                                                        ~\leq~\gamma^2      										& &			 \\
                                                &    \Updownarrow     &      									& &			 \\
                                    \forall ~{w\in \sfL^2},
                                            & &  \gamma^2 \|w\|_2^2~\geq~     \|z\|^2_2 							& &			 \\
                                                &    \Updownarrow     &       									& &			\\
                                    \forall ~{w\in \sfL^2}, 
                                            & & \gamma^2 \|w\|_2^2 -   \|z\|^2_2  ~\geq~ 0 						& & 			\\ 
                                                 &    \Updownarrow     &       									& &			\\
                                    	\displaystyle 
					\inf_{w\in \sfL^2} & &  \gamma^2 \|w\|_2^2 -   \|z\|^2_2  ~=~ 0 
					\hspace{10em}  \mbox{(since with $w=0$, $z=0$)}
                    \end{array}    
             \]
            The above implications, though simple, are remarkable! They
             convert the problem of checking an induced norm
            inequality (a ratio of norms) to one of checking the positivity of a quadratic form over the input and output 
            signals of a system. 
            This problem fits in the general framework of~\req{LQP} because the  
            form is  a  quadratic form  jointly on the input $w$ and the state $x$ (for notational simplicity, we 
            set the direct-feedthrough term $D=0$ in the sequel)
            \begin{align} 
            	\qform_\rmb ~&:=~  \gamma^2 \|w\|_2^2 -   \|z\|^2_2 
			~=~ \gamma^2		\int_0^T	w^*(t) w(t) ~dt ~-~ \int_0^T z^*(t) z(t) ~dt 			\nonumber		\\
			&=~		\int_0^T  \lb 	 \gamma^2 w^* w  ~-~   x^* C^* Cx  \rom \rb ~dt 	
			~=~ \int_0^T \!\!\!	\bbm x^* & w^* \ebm 	
						\bbm \sm C^*C  & 0 \\ 0 & \gamma^2 I \ebm 
						\bbm x \\ w \ebm dt 											\nonumber	\\
			&=~ \inprod{	\bbm x \\ w \ebm	}{ \bbm \sm C^*C  & 0 \\ 0 & \gamma^2 I \ebm \bbm x \\ w \ebm 	}_{L^2[0,\sT]} 
			~=:~ \inprod{	\bbm x \\ w \ebm	}{ 	\calQ_\rmb \bbm x \\ w \ebm 	}_{L^2[0,\sT]} ,
																				\label{LQP_Qb_J.eq}
            \end{align} 
            where $\calQ_\rmb$ is the (possibly time-varying) matrix function defined above. The subscript $\rm b$ references 
            the {\em bounded real lemma} which is another name for the criterion of the same problem we are addressing.

            Another important systems theory criterion is that of passivity.
            \begin{definition}
                An $\sfL^2$-stable linear system  $G$
                 is called {\em  passive over $(0,\sT)$} if  all input output pairs $z=Gw$ with $w\in \sfL^2(0,\sT)$ (and zero initial conditions)
                satisfy 
                \be  
                		\inprod{z}{w \rome}_{\sfL^2(0,\ssT)} 
			~=~ 
                		\smint{0}{T} z^*(t)~w(t)~dt ~\geq~0.
                  \label{pass.eqn}       
                  \ee
                  A system passive over $(0,\infty)$ is simply called {\em passive}. 
            \end{definition}
            We can again convert this to a question about the positivity of a quadratic form in input and state when the 
            system $G$ is given by the realization~\req{system_LQP_IQC}
            \begin{align} 
            	\qform_\rmp ~&:=~ \smint{0}{T}  z^*w~dt 
				~=~  \smint{0}{T} (Cx+Dw)^*~w~dt
				~=~ \smint{0}{T} \lb x^* C^*w +w^*D^*w \rom \rb ~dt
					\nonumber		\\
			&=~		
				\frac{1}{2}	\int_0^T 	\bbm x^* & w^* \ebm 	
						\bbm 0  & C^* \\ C & {D+D^*}  \ebm 
						\bbm x \\ w \ebm dt 											
			~=:~ 
			\frac{1}{2}  \inprod{	\bbm x \\ w \ebm	}{ 	\calQ_\rmp \bbm x \\ w \ebm 	}_{\sfL^2(0,T)} .
																				\label{LQP_Qp_J.eq}
            \end{align} 
	Thus checking passivity of a system is equivalent to checking whether the quadratic form~\req{LQP_Qp_J} is positive 
	on all signals that satisfy the system's equations.

\subsection{Redundancy in Covariance Matrices} 							\label{Cov.subsec}

            Let $z$ be a zero mean, continuous-time  random process. Its instantaneous covariance matrix is
            defined as 
            \[
            	\Sigma(t) ~:=~ \expct{z(t) ~z^*\!(t) \romn} , 
            \]
            which is a deterministic, matrix-valued function of $t$. Recall that covariance matrices are always positive $\Sigma(t)\geq 0$, 
            and therefore have  symmetric factorizations
             $\Sigma(t) = U(t) ~U^*(t)$ for some other matrix $U(t)$. 
            The dimensions of any such  factorization has implications for dependencies between the components of the vector $z(t)$. 
            This last statement will of course be only true up to second order statistics of the process $z$. For Gaussian processes, 
            the statement is true without qualification. 
            
            In the following, the dependence on $t$ is suppressed for notational simplicity. Given a factorization 
            $\expct{z ~z^* \romn} = UU^*$, consider 
            another random vector $\zeta := U w$, where $w$ is a zero-mean random vector with uncorrelated 
            components (i.e. $\expct{ww^*}=I$). With this construction, $\zeta$ and $z$ have the same second order statistics
            as can be easily verified
            \[
            	 \expct{\zeta ~\zeta^* \romn}
            	 ~=~ 
            	 \expct{Uw ~(Uw)^* \romn}	 
            	 ~=~ 
            	 \expct{Uw w^*U^* \romn}	 
            	 ~=~ 
            	 U~ \expct{w w^* \romn}U^*	 
            	 ~=~ 
            	 U~ U^*	 
            	 ~=~ \Sigma .
            \]
            Of course if $z$ and $w$ are Gaussian, then this means that $z$ and $\zeta$ have the same distribution.

            We are interested in full-rank factorizations, i.e. 
            \[
            	\Sigma ~=~ 
            	\tallmat{U} \widemat{U^*} ,
            	\hstm \hstm \Sigma\in\bbS_n, ~U\in\R^{n\times r}, ~r\leq n. 
            \]
            where $r=\rank{\Sigma}$, and therefore $U$ has full column rank. If $r<n$, then the original covariance 
            matrix $\Sigma$ is not full rank. This indicates that there are redundancies in the description of the processes
            $z$ and $\zeta$. To see this, partition $\zeta$ and $U$ as follows 
            \[
            	\zeta ~=~ \bbm \zeta_1 \\ \zeta_2 \ebm , 
            	\hstm 
            	U ~=~ \bbm U_1 \\ U_2 \ebm, 
            \]
            so that $U_1$ is square (i.e. in $\R^{r\times r}$). Assume without loss of generality that $U_1$ is 
            invertible\footnote{Since $U$ is full rank, this can always be done by permuting the rows of $U$ so that the first $r$ are
            	linearly independent. This corresponds to permuting  the components of the vectors 
            	$z$ and $\zeta$.}. 
            With this partition, the relation $\zeta = U w$ now has the following implications 
            \be
            	 \bbm \zeta_1 \\ \zeta_2 \ebm
            	 = \bbm U_1 \\ U_2 \ebm w 
            	 = \bbm U_1w \\ U_2 w\ebm 
            	 \hstm \Leftrightarrow \hstm 
            	 w = U_1^{-1} \zeta_1 
            	 ~~\mbox{and}~ ~
            	 \zeta_2 = U_2 w = U_2 U_1^{-1} \zeta_1. 
              \label{zeta_depend.eq}
            \ee
            Thus $\zeta_2$ is a {\em function} of $\zeta_1$, or equivalently, the entire $\zeta$ vector is completely determined by 
            its subcomponents $\zeta_1$, which is an $r$-vector. In an intuitive sense, the number of ``degrees of freedom'' in the 
            random $n$-vector $\zeta$ is actually $r\leq n$, the rank of its covariance matrix. 
            
            Finally, if $z$ and $w$ are Gaussian, then the components of  $z$ have the same joint distribution as the 
            components of $\zeta$  as mentioned earlier. Since the 
            components of $\zeta$ obey the relation~\req{zeta_depend}, then the components of $z$ obey the same relation. 
            The statements made above hold for each $t$, although the rank of the covariance matrix may depend on $t$.


	\subsection{Riccati Comparisons} 										\label{Ric_diff.appen}
	
		Most of the  arguments in this subsection   are largely a reorganization of the 
		classic arguments in~\cite{willems1971least} for the case of finite time horizon. 
		The expression~\req{Ric_diff_will} below  is from~\cite[Lemma 3]{willems1971least}. 
		The expression~\req{Ric_diff} is a slight modification which is needed when the matrix $M$ below
		has no specific definiteness. 
	
		\begin{myenumerate}
		\item 
		Given the linear-quadratic portion of the  Riccati opeator (i.e. without the constant term) 
		\[
			\cR_{\sf lq}(\Lam) ~:=~ A^*\Lam + \Lam A -\Lam M \Lam , 
		\]
		The difference $\cR_{\sf lq}(\Lam_1) - \cR_{\sf lq}(\Lam_2)$ can be expressed in several ways as follows
		\begin{align} 
			\hspace{-1em}
			\cR_{\sf lq}(\Lam_1) - \cR_{\sf lq}(\Lam_2) 
			~&=~ 
			A^* (\Lam_1-\Lam_2) +  (\Lam_1-\Lam_2)  A 
					-\Lam_1 M \Lam_1 + \Lam_2 M \Lam_2 
				 													\nonumber		\\
			&=~
				A^* (\Lam_1-\Lam_2) + (\Lam_1-\Lam_2) A 
						- \Lam_1 M \Lam_1+ \Lam_2 M \Lam_2  
				 											 		\nonumber		\\						
			&\hstm  
				+\Lam_1 M \Lam_1 -\Lam_1 M \Lam_1 
				+\Lam_1 M \Lam_2 -\Lam_1 M \Lam_2 
				+\Lam_2 M \Lam_1 -\Lam_2 M \Lam_1 				
				 													\nonumber		\\
			 &=~	 \lb A-M\Lam_1 \rb^* (\Lam_1-\Lam_2) + (\Lam_1-\Lam_2) \lb A-M\Lam_1  \rb  
				+   (\Lam_1-\Lam_2) M (\Lam_1-\Lam_2) 
		 																	 \label{Ric_diff_will.eq}	\\
			&=~  
				\big(  A+M(\Lam_1+\Lam_2)/2  \big) ^*~ (\Lam_1-\Lam_2) 
			 			~+~ (\Lam_1-\Lam_2) ~\big(  A+M(\Lam_1+\Lam_2)/2)  \big)  	 \label{Ric_diff.eq}
		\end{align} 
				
%
%

	\item 
	Consider  solutions of two  Riccati differential equations with different ``$Q$-terms''
	\be
	\begin{aligned} 
		0 ~&=~ \Lambdad_1 + A^*\Lambda_1 + \Lambda_1 A - \Lambda_1 M \Lambda_1 +Q_1 	,
							\hstm \Lambda_1 (\sT) =0,									\\
		0 ~&=~ \Lambdad_2 + A^*\Lambda_2 + \Lambda_2 A - \Lambda_2 M \Lambda_2 +Q_2	,
							\hstm \Lambda_2 (\sT) =0,
	\end{aligned} 
	\label{two_Ric_comp.eq}
	\ee
	with $Q_1\geq Q_2$. 
	Define the difference $\Lambdat := \Lambda_1 - \Lambda_2$, subtract the two equations and apply the 
	Riccati difference formula~\req{Ric_diff} to obtain 
	\begin{align*} 
		\dot{\Lambdat}
			+ \lb A+M\big( \Lambda_1+\Lambda_2 \big) /2 \rb^* \Lambdat  
			+ \Lambdat \lb A+M\big( \Lambda_1+\Lambda_2 \big) /2 \rb
			+ \lb Q_1 - Q_2 \rb    
		~&=~ 0 														\\
		\Rightarrow \hstm 
		\dot{\Lambdat}
			+ \lb A+M\big( \Lambda_1+\Lambda_2 \big) /2 \rb^* \Lambdat  
			+ \Lambdat \lb A+M\big( \Lambda_1+\Lambda_2 \big) /2 \rb
		~&=~ - \lb Q_1 - Q_2 \rb    
		~\leq~ 0 . 
	\end{align*} 
	This is a differential Lyapunov inequality of the form~\req{Lyap_ineq_def} for $\Lambdat$. 
	Applying~\req{Lyap_diff_ineq_signs} we conclude that $\Lambdat(t)\geq0$, and  for the
	differential Riccati equations with final conditions~\req{two_Ric_comp}
	\be
		Q_1(t)  \geq Q_2(t), \hsom t\in[0,\sT]  
		\hstm \Rightarrow \hstm 
		\Lambda_1(t) \geq \Lambda_2(t) , \hsom t\in[0,\sT] 
	  \label{Ric_comp_Q1Q2.eq}
	\ee
	
	\item 
	Now consider the DRI~\req{DRI_final} and the DRE obtained from it by setting the inequality to equality
	\be
		\left.
            	\begin{aligned} 
            		\dot{\Lambda} + A^*\Lambda + \Lambda A 
            			- \Lambda~ M~ \Lambda + Q   ~\geq~ 0, 
            		\hstm \Lambda(\sT)=0,											\\
            		\dot{\Lambdab} + A^*\Lambdab + \Lambdab A 
            			- \Lambdab~ M~ \Lambdab + Q   ~=~ 0, 
            		\hstm \Lambdab(\sT)=0.									
            	\end{aligned} 
		\right\} 
		\hstm \Rightarrow \hstm 
		\Lambdab(t)\geq\Lambda(t),
	  \label{DRE_deriv.eq}
	\ee
	where the conclusion  $\Lambdab(t)\geq\Lambda(t)$ follows from previous arguments. Indeed, the DRI is equivalent 
	to the DRE~\req{DRE_H_pos} with the forcing function $H\geq 0$. A comparison of their ``$Q$-terms'' shows
	that $Q-(Q-H)=H\geq0$, and therefore by the comparison result~\req{Ric_comp_Q1Q2} 
	$\Lambdab(t)\geq \Lambda(t)$ for all $t\in[0,\sT]$, i.e. $\Lambdab$ is the {\em maximal solution} to the DRI.

	\end{myenumerate}

	\subsection{DLMIs with initial or final conditions} 						\label{DLMI_i_f.sec}
		
                    	\begin{theorem}[Extremal Solutions] 								\label{DRI_DRE_app.lemma}
                    		Consider the linear matrix $\calM$ and the associated  Riccati  $\cR$ operators
                    		\begin{align} 
                    			\calM(\Lambda)
                    				 ~&:=~ 	 \calQ +\cE^\adj(\Lambdad) + \cA^\adj(\Lambda) 
                    				 ~= 
                    				\bbm 	Q+\dot{\Lambda}+A^*\Lambda+\Lambda A & N+\Lambda B \\ 
                    						 (N+\Lambda B)^*  						& R \ebm  , 		
																						\\
                    			\cR(\Lambda) ~&:=~ \Lambdad +
                    			A^*\Lambda + \Lambda A - \lb N+ \Lambda B\rb R^{\sm1} \lb N+ \Lambda B\rb^*	+ Q,
                    													  \label{Ric_op_def_app.eq}
                    		\end{align} 
				where $R>0$ or $R<0$. 
                    		Consider also initial- and final-value Differential Linear Matrix Inequalities (DLMIs) and their  equivalent 
                    		Differential Riccati Inequalities (DRIs) over $[t_\rmi,t_\rmf]$ 
                    		\be
                    		\begin{aligned}
					\arraycolsep=2pt
					\left. 
					\begin{array}{rclcrcl} 
									\calM(\Lambda) &\geq &  0 & \hstm \Leftrightarrow 	& \hstm\cR(\Lambda) &\geq &  0 \\
					 	\mbox{or}~ 	\calM(\Lambda) &\leq &  0 & \hstm \Leftrightarrow 	& \hstm \cR(\Lambda) &\leq &  0
					\end{array}  \right\} 
                   			&\hstm &			\Lambda(t_\rmi)&=\Lambda_\rmi , 	\hstm {\rm (DLMI_i/DRI_i)}  	\\					
					\arraycolsep=2pt
					\left. 
					\begin{array}{rclcrcl} 
									\calM(\Lambda) &\geq &  0 & \hstm \Leftrightarrow 	& \hstm\cR(\Lambda) &\geq &  0 \\
					 	\mbox{or}~ 	\calM(\Lambda) &\leq &  0 & \hstm \Leftrightarrow 	& \hstm \cR(\Lambda) &\leq &  0
					\end{array}  \right\} 
                    			&\hstm &				\Lambda(t_\rmf)&=\Lambda_\rmf  .	\hstm {\rm (DLMI_f/DRI_f)}  	
                    		\end{aligned}  
                    															\label{DRI_LQR_app.eq}
                    		\ee
                    		Minimal and maximal solutions of the initial- and final-value inequalities respectively are given by the 
                    		following Differential Riccati Equations (DREs) 
                    		\begin{align}
                    			 \cR(\Lambdaub) 
                    			 	~&=~0 , 
                    					&&  \Lambdaub(t_\rmi)=\Lambda_\rmi ,  \hstm	{\rm (DRE_i):} ~\mbox{\rm minimal-initial} 		
                    																		\label{DRE_LQRi_app.eq}	\\
                    			   \cR(\Lambdab) 
                    			 	~&=~0 , 
                    					&&  \Lambdab(t_\rmf)=\Lambda_\rmf  , \hstm	{\rm (DRE_f):} 	~\mbox{\rm maximal-final} 	
                    																		\label{DRE_LQRf_app.eq}
                    		\end{align} 
                    		i.e. for any $\Lambda$  satisfying 
                    		$\rm DLMI_i/DRI_i$ (resp. $\rm DLMI_f/DRI_f$)
                    		\[ 
                    			 \Lambdaub(t) ~\leq~ \Lambda(t) 
                    			 \hstm 		 
                    			\lb\rom  \mbox{resp.} ~  \Lambda(t) ~\leq~\Lambdab(t) \rb , 			 			 
                    			 \hstm\hstm t\in[t_\rmi,t_\rmf] . 
                    		\]
                    		Furthermore, at the extremal solutions, the matrix $\calM$ has the 
                    		full-rank symmetric factorization 
                    		\be
                    			\calM(\Lambdab) 
                    			~=~ 
                    			\bbm (N+\Lambdab B) R^{\sm\frac{1}{2}}  \\  R^{\frac{1}{2}} \ebm  
                    			\bbm  R^{\sm\frac{1}{2}} (N+ \Lambdab B)^* &  R^{\frac{1}{2}} \ebm,
                    		  \label{DLMI_rank_min_app.eq}
                    		\ee
                    		and similarly for $\calM(\Lambdaub)$. 
                    	\end{theorem}
		
	\begin{proof} 
		We consider the case $\calM(\Lambda)\geq0$ with $\Lambda(t_\rmf)=\Lambda_\rmf$. The other three cases are
		argued similarly. 		The Riccati operator~\req{Ric_op_def_app} can be written in the more standard form 
		\begin{align} 
			\cR(\Lambda) ~&:=~ 
			\Lambdad + A^*\Lambda + \Lambda A - \lb N+ \Lambda B\rb R^{\sm1} \lb N+ \Lambda B\rb^*	+ Q		\nonumber	\\
			~&=~\Lambdad +
				\big( A - BR^{\sm1}N^* \big)^* ~\Lam + \Lam ~\big( A - BR^{\sm1}N^* \big) 
				- \Lam~ B R^{\sm1} B^*~ \Lam + \big( Q-NR^{\sm1} N^* \big) 						\nonumber	\\
			&=:~ \Lambdad + \Ah^* \Lambda + \Lambda \Ah - \Lam M \Lam + \Qh .
																	\label{DRE_standard_f.eq}
		\end{align} 
		The 
		 Riccati comparison result~\req{DRE_deriv}  now implies that the solution $\Lambdab$ of $\cR(\Lambdab)=0$
		  is maximal over all
		 solutions of the $\rm DRE_f$, i.e. over all solutions of the $\rm DLMI_f$ 

		For the factorization~\req{DLMI_rank_min}, note that the $\rm DRE_f$ for  $\Lambdab$ implies  
            	\[
            		Q+\dot{\Lambdab} + A^*\Lambdab + \Lambdab A
            		~=~		(N+\Lambdab B) R^{\sm1} (N+\Lambdab B)^*.
            	\] 
            	 Using this, the matrix $\calM(\Lambdab)$  can be decomposed as
            	\begin{align*}
                        	  \hspace*{-1em} 
            		\bbm Q+\dot{\Lambdab} + A^*\Lambdab + \Lambdab A & (N+\Lambdab B) \\  (N+\Lambdab B)^* &  R \ebm 
            		&= 
            		\bbm (N+\Lambdab B) R^{\sm1} (N+\Lambdab B)^* & N+\Lambdab B \\ (N+\Lambdab B)^* &  R \ebm 	\nonumber	\\
            		&= 
            		\bbm (N+\Lambdab B) R^{\sm\frac{1}{2}}  \\  R^{\frac{1}{2}} \ebm  
            			\bbm  R^{\sm\frac{1}{2}} (N+ \Lambdab B)^* &  R^{\frac{1}{2}} \ebm , 
		\end{align*} 
		which is a full-rank factorization since $R$ (and therefore $R^{\frac{1}{2}}$) is assumed non-singular. 		
	\end{proof}

\bibliographystyle{IEEEtran}
\bibliography{LQP}

\end{document}

%% file: common_commands.tex
    \usepackage{amsmath,amssymb,amsfonts,amsbsy,amsthm}
    \usepackage{latexsym,graphics,graphicx,color,sectsty,mathdots,bm}
    
    \usepackage{subcaption}
    \usepackage{layout}
    \usepackage{bbm}
    \usepackage[colorlinks=true,linkcolor=blue]{hyperref}
    \usepackage{calc,cancel}
    \usepackage{enumitem}
    \usepackage{rotating}
    \usepackage[table]{xcolor}
    \usepackage{cleveref}
    \usepackage{appendix} 
    \usepackage{sidecap}
    \usepackage{amscd}
    \usepackage{mathtools}
    \usepackage{scalerel}

    \usepackage{arydshln}
    \usepackage{makeidx} 
    \usepackage{nicematrix}
    \usepackage{fancyhdr}
    \usepackage{xfrac}
    \usepackage{tocvsec2}

\theoremstyle{plain}
\newtheorem{theorem}{Theorem}

\newtheorem{lemma}[theorem]{Lemma}
\newtheorem*{lemma*}{Lemma}

\newtheorem{definition}[theorem]{Definition}

\newtheorem*{summary*}{Summary}

\theoremstyle{definition}
\newtheorem{example}[theorem]{Example}

\theoremstyle{remark}

\newcommand{\trcc}[1]   { {\rm tr}\!\left( #1  \right) }

	\definecolor{bgblue}{rgb}{0.04,0.39,0.53}
	\definecolor{dblue}{rgb}{0,0.3,0.7}
	\definecolor{ddblue}{rgb}{0,0.1,0.6}
	\definecolor{ddgreen}{rgb}{0,0.25,0.05}
	\definecolor{dgreen}{rgb}{0,0.5,0.05}

\newcommand{\enma}[1]   {\ensuremath{#1}}

\newcommand{\req}[1]{(\ref{#1.eq})}

\newcommand{\beq}{\begin{equation}}
\newcommand{\eeq}{\end{equation}}
\newcommand{\beqn}{\begin{eqnarray}}
\newcommand{\eeqn}{\end{eqnarray}}
\newcommand{\beqns}{\begin{eqnarray*}}
\newcommand{\eeqns}{\end{eqnarray*}}
\newcommand{\bct}{\begin{center}}
\newcommand{\ect}{\end{center}}
\newcommand{\btmz}{\begin{itemize}}
\newcommand{\etmz}{\end{itemize}}
\newcommand{\benum}{\begin{enumerate}}
\newcommand{\eenum}{\end{enumerate}}

\newcommand{\R}{{\mathbb R}}

\newcommand{\E}{{\mathbb E}}

\newcommand{\bbS}{{\mathbb S}}

\newcommand{\sT}{{\scriptstyle T }}

\newcommand{\cL}{\enma{\mathcal L}}

\newcommand{\cR}{\enma{\mathcal R}}

\newcommand{\Hinfty}{{{\sf H}^{\infty} } }

\newcommand{\diag}[1]{ {\sf diag} \! \left( #1 \right) \rule{0em}{1em}}

\newcommand{\Ah}{\hat{A}}

\newcommand{\Qh}{\hat{Q}}

\newcommand{\ub}{{\bar{u}}}

\newcommand{\xb}{{\bar{x}}}

\newcommand{\lambdab}{{\bar{\lambda}}}
\newcommand{\Lambdat}{{\tilde{\Lambda}}}

\newcommand{\Lambdaub}{{\underline{\Lambda}}}
\newcommand{\Lambdab}{{\bar{\Lambda}}}
\newcommand{\Lambdad}{{\dot{\Lambda}}}
\newcommand{\Sigmab}{{\bar{\Sigma}}}

\newcommand{\Lambdablue}{{\color{blue}{\Lambda}}}

\newcommand{\bbm}{\begin{bmatrix}} 
\newcommand{\ebm}{\end{bmatrix}} 

\newcommand{\bsm}{\left[ \begin{smallmatrix}} 
\newcommand{\esm}{\end{smallmatrix} \right]} 

\newcommand{\bsbm}{\left[ \begin{smallbmatrix}} 
\newcommand{\esbm}{\end{smallbmatrix} \right]} 

\newcommand{\bbNm}{\begin{bNiceMatrix}} 				
\newcommand{\ebNm}{\end{bNiceMatrix}} 
\newcommand{\bNA}[1]{ \left[ \begin{NiceArray}{#1} } 		
\newcommand{\eNA}{ \end{NiceArray} \right] }

\newcommand{\lb}{\left(}
\newcommand{\rb}{\right)}
\newcommand{\lcb}{\left\{}
\newcommand{\rcb}{\right\}}

\newcommand{\calM}{{\cal M}}
\newcommand{\calQ}{{\cal Q}}

\newcommand{\sfH}{{\sf H}}
\newcommand{\sfP}{{\sf P}}

\newcommand{\cE}{{\mathcal E}}

\newcommand{\be}{\begin{equation}}
\newcommand{\ee}{\end{equation}}

\newcommand{\Ltwo}{\sfL^2}

\newcommand{\cplxs}{ C\kern -.35em \rule{0.03 em}{.7 ex}~   }

\def\complex{\hbox{C\kern -.45em \rule{0.03 em}{1.5 ex}}~}

\newcommand{\Gh}{\hat{G}}

\newcommand{\rmb}{{\mathrm b}}

\newcommand{\rmi}{{\rm i}} 
\newcommand{\rmf}{{\rm f}} 
 
\newcommand{\rmx}{{\rm x}}

\newcommand{\bi}{\begin{itemize}}
\newcommand{\ei}{\end{itemize}}
\newcommand{\ben}{\begin{enumerate}}
\newcommand{\een}{\end{enumerate}}

\newcommand{\cA}{\mathcal{A}}

\newcommand{\expct}[1]{  \E  \left[ #1 \right]    }

\newcommand{\inprod}[2]{\left< #1 \boldsymbol{,} #2 \right>}

\newcommand{\bseq}{\begin{subequations}}
\newcommand{\eseq}{\end{subequations}}

\newcommand{\bbP}{\mathbb{P}}

\newcommand{\ba}{\begin{array}}
\newcommand{\ea}{\end{array}}

\newcommand{\mycaption}[1]{\caption{\footnotesize #1}}

\definecolor{dred}{rgb}{.8,0,0}

\newcommand{\tcdb}[1]{\textcolor{dblue}{#1}}

\newcommand{\sm}{\text{-}}

\newcommand{\ssT}{{\scriptscriptstyle T}}

\def\clap#1{\hbox to 0pt{\hss#1\hss}}

\newcommand{\btc}{\begin{tabular}{c}}
\newcommand{\btbl}{\begin{tabular}{l}}
\newcommand{\et}{\end{tabular}}

    \newcommand{\Xba}{{\bar{X}}}
    \newcommand{\Yba}{{\bar{Y}}}
    \newcommand{\Zba}{{\bar{Z}}}

    \newcommand{\vb}{{\bar{v}}}

    \newcommand{\Xd}{\dot{X}}
    \newcommand{\xd}{\dot{x}}

	\newcommand{\rom}{\rule{0em}{1em}}
	
	\newcommand{\romn}{\rule{0em}{.91em}}
	\newcommand{\rome}{\rule{0em}{.85em}}

\newcommand{\Ims}[1]{{\sf Im}\!\left( #1 \right) }

\newcommand{\rank}[1]{{\rm \bf rk} \!\lb #1 \rb}
\newcommand{\ranks}[1]{{\rm \bf rk} \lb #1 \rb}

\newcommand{\hsom}{\hspace{1em}} 
\newcommand{\hstm}{\hspace{2em}}

\newcommand{\heavi}{\mathfrak{h}}

\newcommand{\rme}{{\rm e}}

\newcommand{\rmp}{{\rm p}}

\newcommand{\rmo}{{\rm o}}

\newcommand{\sfC}{{\sf C}} 

\newcommand{\sfV}{{\sf V}}

\newcommand{\sfW}{{\sf W}}
\newcommand{\sfX}{{\sf X}}

\newenvironment{myenumerate}{\begin{enumerate}[left=0.2em]}{\end{enumerate}}

\newcommand{\sfS}{{\sf S}}

\newcommand{\sfL}{{\sf L}}

\newcommand{\sfAC}{{\sf AC}}

\newcommand{\Conv}[1]{{\sf Conv}\!\lb #1 \rb }

\newcommand{\Sigmad}{\dot{\Sigma}}

\newcommand{\Em}{\bbm I & 0 \ebm} 
\newcommand{\Ems}{\bbm I \\ 0 \ebm}

\newcommand{\intf}{{\sf {\bm I}}_{\sf  f}}
\newcommand{\intb}{{\sf {\bm I}}_{\sf \bm b}}

\newcommand{\smint}[2]{\scaleobj{.8}{\int_{{#1}}^{{#2}}}}

\newcommand{\blb}{\big(}

\newcommand{\lbb}{\big(}
\newcommand{\rbb}{\big)}

\newcommand{\sigmax}{{\bar{\sigma}}} 

\newcommand{\qform}{{\mbox{\large $\mathbf q$}}}
\newcommand{\sfq}{{\sf q}}

\newcommand{\bigmat}[1]{{\arraycolsep=3pt \def\arraystretch{.5} \bbm & & \\ & #1 & \\ & & \ebm }}
\newcommand{\tallmat}[1]{{ \def\arraystretch{.5} \bbm \rom  \\  #1  \\  \rom \ebm }}
\newcommand{\widemat}[1]{{\arraycolsep=2pt  \bbm ~& #1 &~  \ebm }}

\DeclareMathOperator*{\ninf}{-~inf}

\newcommand{\Lam}{\Lambda} 

\newcommand{\scpt}{\scriptstyle}

	\newcommand{\bbms}{\begin{bsmallmatrix}}
	\newcommand{\ebms}{\end{bsmallmatrix}}

               \DeclareMathAlphabet{\mymathbb}{U}{BOONDOX-ds}{m}{n}

	\newcommand{\LOne}{{\sfL}^{\! 1}}
	\newcommand{\LTwo}{{\sfL}^{\! 2}}
	\newcommand{\LInf}{{\sfL}^{\!\! \infty}}

	\newcommand{\adj}{{ {\bm *}}}

		\newcommand{\deffont}[1]{{\sf #1}}